\documentclass[a4paper,english]{lipics-derivative-centered}

\usepackage{amsmath,amssymb,amsthm}
\usepackage[fleqn]{mathtools} 
\usepackage{enumitem}
\usepackage{xcolor}
\usepackage{thmtools}		
\usepackage{thm-restate} 	
\usepackage{hyperref}
\usepackage{cite} 			
\usepackage{graphicx}		

\newcommand{\pls}{\mathsf{size}\text{-}\mathsf{pls}}
\newcommand{\size}{\mathsf{size}}
\DeclareMathOperator{\poly}{poly}
\DeclareMathOperator{\dist}{\textrm{dist}}
\DeclareMathOperator{\id}{\textsc{ID}}
\newcommand{\LOCAL}{\textsc{local}}
\newcommand{\CONGEST}{\textsc{congest}}
\newcommand{\radius}{\mathsf{radius}}

\newcommand{\diam}{{\sc diam}}
\newcommand{\mdiam}{\mbox{\diam}}

\newcommand{\ST}{\textsc{st}}
\newcommand{\MST}{\textsc{mst}}
\newcommand{\POS}{\mbox{pos}}

\newcommand{\disj}{{\sc disj}}

\newcommand{\set}[1]{\left\{ #1 \right\}}
\DeclarePairedDelimiter\ceil{\lceil}{\rceil}

\newcommand{\cF}{\mathcal{F}}
\newcommand{\ID}{\mbox{ID}}

\newtheorem{problem}{Open Problem}

\declaretheorem[name=Claim]{claim}

\newcommand{\remove}[1]{}

\title{Redundancy in Distributed Proofs}

\author[1]{Laurent Feuilloley}
\affil[1]{Universidad de Chile, Chile. {\tt feuilloley@dii.uchile.cl}}

\author[2]{Pierre Fraigniaud}
\affil[2]{IRIF, CNRS and Université de Paris, France. {\tt pierref@irif.fr}}

\author[3]{Juho Hirvonen}
\affil[3]{Aalto University, Finland. {\tt juho.hirvonen@aalto.fi}}

\author[4]{Ami Paz}
\affil[4]{Faculty of Computer Science, University of Vienna, Austria. {\tt ami.paz@univie.ac.at}}

\author[5]{Mor Perry}
\affil[5]{Weizmann Institute of Science, Israel. {\tt mor.perry@weizmann.ac.il}}

\authorrunning{Laurent Feuilloley, Pierre Fraigniaud, Juho Hirvonen, Ami Paz, and Mor Perry}

\Copyright{Laurent Feuilloley, Pierre Fraigniaud, Juho Hirvonen, Ami Paz, and Mor Perry}

\subjclass{Networks---Error detection and error correction;
		Theory of computation---Distributed computing models;
		Computer systems organization---Redundancy}


\begin{document}

\maketitle

\begin{abstract}
Distributed proofs are mechanisms that enable the nodes of a network to collectively and efficiently check the correctness of Boolean predicates
on the structure of the network (e.g., having a specific diameter),
or on objects distributed over the nodes (e.g., a spanning tree).
We consider well known mechanisms consisting of two components:
a \emph{prover} that assigns a \emph{certificate} to each node,
and a distributed algorithm called a \emph{verifier} that is in charge of verifying the distributed proof formed by the collection of all certificates.
We show that many network predicates have distributed proofs offering a high level of redundancy, explicitly or implicitly.
We use this remarkable property of distributed proofs to establish perfect tradeoffs between the \emph{size of the certificate} stored at every node, and the \emph{number of rounds} of the verification protocol.
\end{abstract}

\newpage

\section{Introduction}

\subsection{Context and Objective}

In the context of distributed fault-tolerant computing in large scale networks, it is of the utmost importance that the computing nodes can perpetually check the correctness of distributed objects (e.g., spanning trees) encoded distributively over the network. Such objects can be the outcome of an algorithm that might be subject to failures, or be a-priori correctly given objects but subject to later corruption. 
There are several mechanisms for checking the correctness of distributed objects (see, e.g., \cite{afek1997,AfekD02,AwerbuchPV91,BeauquierDDT07,BlinF15,BlinFP14}), and here we focus on one classical mechanism which is both simple and versatile, known as \emph{proof-labeling schemes}~\cite{KormanKP10}, or as~\emph{locally checkable proofs}~\cite{GoosS16}\footnote{These two mechanisms slightly differ: the latter assumes that every node can have access to the whole state of each of its neighbors, while the former assumes that only part of this state (the certificate that is assigned for the verification process) is visible from neighboring nodes; nevertheless, the two mechanisms share the same essential features.}.

Roughly, a proof-labeling scheme assigns \emph{certificates} to each node of the network. These certificates can be viewed as forming a distributed proof of the actual object.
The nodes are then in charge of collectively verifying the correctness of this proof. The requirements are in a way similar to those imposed on non-deterministic algorithms (e.g., the class \textsf{NP}), namely: (1)~on correct structures, the assigned certificates must be accepted, in the sense that every node must accept its given certificate; (2)~on corrupted structures, whatever certificates are given to the nodes, they must be rejected, in the sense that at least one node must reject its given certificate. (The rejecting node(s) can raise an alarm, or launch a recovery procedure.) 
For example, a spanning-tree can be verified in one communication round by the following process. Root the tree and give every node a certificate which is the identity of the root and its distance from the root. 
The nodes use the distances to verify that the structure is acyclic, by making sure each non-root node has a neighbor with smaller distance;
they use the root-$\ID$ to verify connectivity, by making sure they all have the same root-$\ID$.
If the structure is connected and acyclic, it is a spanning tree.
Proof-labeling schemes and locally checkable proofs can be viewed as a form of non-deterministic distributed computing (see also~\cite{FraigniaudKP13}).

The main measure of quality for a proof-labeling scheme is the \emph{size} of the certificates assigned to \emph{correct} (a.k.a.~\emph{legal}) objects. These certificates are verified using protocols that exchange certificates between neighboring nodes, so using large certificates may result in significant overheads in term of communication. In addition, proof-labeling schemes can be combined with other  mechanisms enforcing fault-tolerance, including replication (multiple copies of the memory of each node). Large certificates may prevent replication, or at the least result in significant overheads in term of space complexity if using replication.

Proof-labeling schemes are extremely versatile, in the sense that they can be used to certify \emph{any} distributed object or graph property. For instance, to certify a spanning tree, there are several proof-labeling schemes, each using certificates of logarithmic size~\cite{ItkisL94,KormanKP10}. 
Similarly, certifying a minimum-weight spanning tree (MST) can be achieved with certificates of size $\Theta(\log^2n)$ bits in $n$-node networks~\cite{KormanKP10, KormanK07}. Moreover, proof-labeling schemes are very \emph{local}, in the sense that the verification procedure performs in just one round of communication, with each node accepting or rejecting based solely on its certificate and the certificates of its neighbors. However, this versatility and locality comes with a cost. 
For instance, certifying rather simple graph property, such as certifying that each node holds  the value of the diameter of the network, requires certificates of $\widetilde{\Omega}(n)$ bits~\cite{Censor-HillelPP18}%
\footnote{
	The $O$-tilde notation $\widetilde{O}(f(n))$ is similar to the big-$O$ notation $O(f(n))$, but ignores factors which are poly-logarithmic in~$f(n)$.
}.  There are properties that require even larger certificates. For instance, certifying that the network is not 3-colorable, or certifying that the network has a non-trivial automorphism both require certificates of $\widetilde{\Omega}(n^2)$ bits~\cite{GoosS16}. The good news though is that all distributed objects (and graph properties) can be certified using certificates of $O(n^2+kn)$ bits, which is the number of bits needed to store the entire $n$-node graph with $k$-bit node labels --- see~\cite{GoosS16,KormanKP10}.

Several attempts have been made to make proof-labeling schemes more efficient.
For instance, it was shown in~\cite{FraigniaudPP19} that randomization helps a lot in terms of \emph{communication} costs, typically by hashing the certificates,
but this might actually come at the price of dramatically increasing the certificate size.
Sophisticated deterministic and efficient solutions have also been provided for reducing the size of the certificates, but they are targeting specific structures only, such as MST~\cite{KormanKM15}.
Another direction for reducing the size of the certificates consists of relaxing the decision mechanism to be some function of the local outputs (not necessarily a conjunction), and allowing each node to output more than just a single bit (accept or reject)~\cite{ArfaouiFIM,ArfaouiFP13}. 
For instance, certifying cycle-freeness simply requires certificates of $O(1)$ bits with just 2-bit output, while certifying cycle-freeness requires certificates of $\Omega(\log n)$ bits with 1-bit output~\cite{KormanKP10}. However, this relaxation assumes the existence of a centralized entity gathering the outputs from the nodes, and there are still network predicates that require certificates of $\widetilde{\Omega}(n^2)$ bits even under this relaxation.
Another notable approach is using approximation~\cite{Censor-HillelPP18},
which reduces, e.g., the certificate size for certifying the diameter of the graph
from $\Omega(n)$ down to $O(\log n)$, but at the cost of only determining if the given value is up to two times the real diameter.

In this paper, we aim at designing deterministic and generic ways for reducing the certificate size of proof-labeling schemes. This is achieved by following the guidelines of \cite{OstrovskyPR17}, that is, trading time for space by exploiting the  inherent redundancy in distributed proofs.
We focus only on questions regarding the \emph{existence} of proof-labeling schemes, and give upper and lower bounds on their sizes. 
While some of our techniques, such as randomized proofs of existence of some efficient schemes, seem to point to specific ways of implementation, we leave the design for efficient algorithms for constructing the labels as an intriguing open question.

\subsection{Our Results}

As mentioned above, proof-labeling schemes include a verification procedure consisting of a single round of communication. In a nutshell, we prove that using more rounds of communication for verifying the distributed proof enables one to reduce significantly the size of the certificates, often by a  factor super-linear in the number of rounds, and sometimes even exponential.

More specifically, a proof-labeling scheme of radius~$t$ (where $t$ can depend on the size of the input graph) is a proof-labeling scheme where the verification procedure performs $t$ rounds, instead of just one round as in classical proof-labeling schemes. We may expect that proof-labeling schemes of radius~$t$ should help reduce the size of the certificates.
This expectation is based on the intuition that the verification of classical (radius-1) proof-labeling schemes is done by comparing certificates of neighboring nodes or computing some function of them,
and accepting only if they are consistent with one another (in a sense that depends on the scheme).
If the certificates are poorly correlated, then allowing more rounds for the verification should not be of much help as, with a $k$-bit certificate per node, the global proof has $kn$ bits in total in $n$-node graphs, leaving little freedom for reorganizing the assignment of these $kn$ bits to the $n$ nodes. 
Perhaps surprisingly, we show that distributed proofs do not only involve partially redundant certificates, but inherently \emph{highly redundant certificates}, which enables one to reduce their size significantly when more rounds are allowed.
To capture this phenomenon, we say that a proof-labeling scheme \emph{scales} with scaling factor $f(t)$ if its size can be reduced by a factor~$\Omega(f(t))$ when using a $t$-round verification procedure;
we say that the scheme \emph{weakly} scales with scaling factor~$f(t)$ if
the scaling factor is $\Omega(f(t)/\poly\log\, n)$ in $n$-node networks.

Our results can be split into three classes, as follows.

\subparagraph{All schemes, specific graphs.}
In Section~\ref{sec:scaling-on-trees}, we prove that, in trees and other graph classes including e.g. grids, \emph{all} proof-labeling schemes scale, with scaling factor~$t$ for $t$-round verification procedures. In other words, for every Boolean predicate $\mathcal{P}$ on labeled trees (that is, trees whose every node is assigned a label, i.e., a binary string), if $\mathcal{P}$ has a proof-labeling scheme with certificates of $k$ bits, for some $k\geq 0$, then $\mathcal{P}$ has a proof-labeling scheme of radius~$t$ with certificates of $O(k/t)$ bits, for all $t\geq 1$. See Theorem~\ref{thm:tree-scaling} and Corolary~\ref{cor:cycle-and-grid-scaling}.

\subparagraph{Specific families of schemes, all graphs.}
While the above results are restricted to specific families of graphs, in Section~\ref{sec:optimal-uniform} we study general graphs, restricting our attention to specific families of proof-labeling schemes.
We prove that, in any graph, uniform parts of proof-labeling schemes weakly scale in an growth-dependent manner, which could be much faster than linear. More precisely, for every Boolean predicate $\mathcal{P}$ on labeled graphs, if $\mathcal{P}$ has a proof-labeling scheme such that $k$ bits are identical in all certificates, then the part with these $k$ bits weakly scales in proportion to the number of nodes in a ball in the graph: it can be reduced into roughly $\widetilde O(k/b(t))$ bits by using a proof-labeling scheme of radius~$t$, where $b(t)$ denotes the size of the smallest ball of radius $t$ in the actual graph. Therefore, in graphs whose neighborhoods increase polynomially, or even exponentially with their radius, the benefit in terms of space-complexity of using a proof-labeling scheme with radius~$t$ can be enormous. This result is of particular interest for the so-called \emph{universal} proof-labeling scheme, in which every node is given the full $n^2$-bit adjacency matrix of the graph as part of its certificate, along with the $O(\log n)$-bit index of that node in the matrix. See Theorem~\ref{theo:universal} and Corollary~\ref{cor:universal}.

\subparagraph{Specific predicates, all graphs.}
We complement these results of scaling for general predicates with results regarding scaling for some concrete predicates.
We address classical Boolean predicates on labeled graphs, including spanning tree, minimum-weight spanning tree, diameter, and additive spanners. 
In Section~\ref{sec:distance-related}, we show that the certificate sizes of proof-labeling schemes for diameter and spanners predicates weakly scale linearly on general graphs, and cannot be scaled better. See Theorems~\ref{theo:scaling for diameter} and~\ref{thm: spanners scaling}.
In Section~\ref{sec: spanning trees}, we study the classical predicates of spanning tree and minimum spanning tree --- for both, we present proof-labeling schemes that scale linearly, and in the minimum spanning tree case, this bridges the gap between two known schemes: a distance-1 $O(\log^2n)$-bit scheme~\cite{KormanKM15}, and a distance-$O(\log n)$ $O(\log n)$-bit scheme~\cite{KormanKM15,KormanK07}. See Theorem~\ref{thm:ST} and~\ref{thm:MST}.

\subsection{Our Techniques}

Our proof-labeling schemes demonstrate that if we allow $t$ rounds of verification,
it is enough to keep only a small portion of the certificates,
while all the rest are redundant.
In a path, it is enough to keep only two consecutive certificates out of every $t$:
two nodes with $t-2$ missing certificates between them can
try all the possible assignments for the missing certificates,
and accept only if such an assignment exists.
This reduces the \emph{average} certificate size;
to reduce the \emph{maximal} size,
we split the remaining certificates equally among the certificate-less nodes.
This idea is extended to trees and grids,
and is at the heart of the proof-labeling schemes presented in Section~\ref{sec:scaling-on-trees}.

On general graphs, we cannot omit certificates from some nodes and let the others check all the options for missing certificates in a similar manner.
This is because, for our approach to apply, the parts of missing certificates must be isolated
by nodes with certificates. That is, when removing the nodes with certificates from the graph, what remains should consists only of small connected components.
However, if all the certificates are essentially the same,
as in the case of the universal scheme,
we can simply keep each part of the certificate at some random node\footnote{All our proof-labeling schemes are deterministic, but we use the probabilistic method for proving the existence of some of them.},
making sure that each node has all parts in its radius-$t$ neighborhood.
A similar, yet more involved idea,
applies when the certificates are distances,
e.g., when the predicate to check is the diameter,
and the certificate of a node contains in a radius-$1$ proof-labeling scheme its distances to all other nodes.
While the certificates are not universal in this latter case,
we show that it still suffices to randomly keep parts of the distances,
such that on 
a fixed shortest path between every two nodes,
the distance between two certificates kept is at most $t$.
These ideas are applied in Sections~\ref{sec:optimal-uniform} and~\ref{sec:distance-related}.

In order to prove lower bounds on the certificate size of proof-labeling schemes
and on their scaling,
we combine several known techniques in an innovative way.
A classic lower bound technique for proof-labeling schemes is called \emph{crossing},
but this cannot be used for lower bounds higher than logarithmic,
and is not suitable for our model.
A more powerful technique is the use of nondeterministic communication complexity~\cite{GoosS16,Censor-HillelPP18},
which extends the technique used for the \CONGEST{} model~\cite{FrischknechtHW12,AbboudCHK16}.
In these bounds, the  nodes are partitioned between two players,
who simulate the verification procedure in order to solve a communication complexity problem,
and communicate whenever a message is sent over the edges of the cut between their nodes.
When proving lower bounds for proof-labeling schemes,
the nondeterminism is used to define the certificates: a nondeterministic string for a communication complexity problem
can be understood as a certificate,
and, when the players simulate verification on a graph,
they interpret their nondeterministic strings as node certificates.

However, even this technique does not seem to be powerful enough
to prove lower bounds for our model, due to the multiple rounds verification.
When splitting the nodes between the two players,
the first round of verification only depends on the certificates of the nodes touching the cut,
but arguing about the other verification rounds seems much harder.
To overcome this problem,
we use a different style of simulation argument,
where the node partition is not fixed but evolves over time~\cite{PelegR99,Elkin04,DasSarma+12}.
More specifically, while there are sets of nodes which are simulated explicitly by either of the two players during the $t$~rounds,
the nodes in the paths connecting these sets are simulated in a decremental manner:
both players start by simulating all these nodes, and then
simulate less and less nodes as time passes.
After the players communicate the certificates of the nodes along the paths at the beginning,
they can simulate the verification process without any further communication.
In this way,
we are able to adapt some techniques used for the \CONGEST{} model (that is a model where one has to cope with bandwidth restrictions) to our model,
even though proof-labeling schemes are a computing model that is much more similar to the \LOCAL{} model~\cite{Peleg00} (where there is no bandwidth restriction).

\subsection{Previous Work}

The mechanism considered in this paper for certifying distributed objects and predicates on labeled graphs has at least three variants. The original \emph{proof-labeling schemes}, as defined in \cite{KormanKP10}, assume that nodes exchange solely their certificates between neighbors during the verification procedure. Instead, the variant called \emph{locally checkable proofs}~\cite{GoosS16} imposes no restrictions on the type of information that can be exchanged between neighbors during the verification procedure. In fact, they can exchange their full labels and local views of the graph, which makes the design of lower bounds far more complex. There is a third variant, called \emph{nondeterministic local decision}~\cite{FraigniaudKP13}, which prevents using the actual identities of the nodes in the certificates. That is, the certificate must be oblivious to the assignment of identifiers to the nodes. This latter mechanism is weaker than proof-labeling schemes and locally checkable proofs, as there are graph predicates that cannot be certified in this manner. 
For example, the certificates can not indicate a unique leader without using the identities. Note that if the certificates contain some artificial assignment of identities, the nodes must verify their uniqueness, which in turn, is impossible without using the unique node identifiers.
However, all  predicates on labeled graphs can be certified by allowing randomization~\cite{FraigniaudKP13}, or by allowing just one alternation of quantifiers (the analog of $\Pi_2$ in the polynomial hierarchy)~\cite{BalliuDFO17}. A recent line of work studies a distributed variant of interactive proofs~\cite{kol18interactive,NaorPY20,CrescenziFP19}.

Our work focuses on the second model, of locally checkable proofs. However, for consistency with previous literature on scaling~\cite{OstrovskyPR17}, we stick to the name proof-labeling schemes.
Note that the main difference is the use of node identifiers, and when the certificates are of size $\Omega(\log n)$, the prover can add the idbandwidth restrictionentifiers to the certificates, making the models equivalent.

Our work was inspired by~\cite{OstrovskyPR17}, which aims at reducing the size of the certificates by trading time for space, i.e., allowing the verification procedure to take $t$~rounds, for a non-constant $t$, in order to reduce the certificate size. They show a tradeoff of this kind for example for proving the acyclicity of the input graph.
The results in~\cite{KormanKM15} were another source of inspiration, as it is shown that, by allowing $O(\log^2n)$ rounds of communication, one can verify MST using certificates of $O(\log n)$ bits. In fact, \cite{KormanKM15} even describe an entire self-stabilizing algorithm for MST construction based on this mechanism for verifying MST.

In~\cite{FeuilloleyFH16}, the authors generalized the study of the class log-\textsf{LCP} introduced in~\cite{GoosS16}, consisting of network properties verifiable with certificates of $O(\log n)$ bits, to present a whole local hierarchy inspired by the polynomial hierarchy. For instance, it is shown that MST is at the second level of that hierarchy, and that there are network properties outside the hierarchy. In~\cite{Patt-ShamirP17}, the effect of sending different messages to different neighbors on the communication complexity of verification is analyzed. The impact of the number of errors on the ability to detect the illegality of a object w.r.t.\ a given predicate is studied in~\cite{FeuilloleyF17}. The notion of approximate proof-labeling schemes was investigated in~\cite{Censor-HillelPP18}, and the impact of randomization on communication complexity of verification has been studied in~\cite{FraigniaudPP19}.

Finally, verification mechanisms a la proof-labeling schemes were used in other contexts, including the congested clique~\cite{KorhonenS17}, wait-free computing~\cite{FraigniaudRT16}, failure detectors~\cite{FraigniaudRTKR16}, anonymous networks~\cite{FoersterLSW18}, and mobile computing~\cite{BampasI16,FraigniaudP17}.  For more references to work related to distributed verification, or distributed decision in general, see the survey~\cite{FeuilloleyF16}. To our knowledge, the aforementioned works~\cite{KormanKM15,OstrovskyPR17} are the only prior works studying tradeoffs between time and certificate size.

\section{Model and Notations}
\label{sec:model-and-notations}

A \emph{distributed network} is modeled as a graph composed of nodes connected by edges, where the nodes represent processors and each edge represents a communication link. Each node has a unique identity number (henceforth, $\ID$). The edges are undirected and the communication on each edge is bidirectional. The nodes communicate synchronously over the edges, i.e., a computation starts simultaneously in all nodes and proceeds in discrete rounds. In each round, each node sends messages to its neighbors, receives messages from its neighbors, and performs a local computation. Nodes are allowed to perform arbitrarily heavy local computations, and these are not taken into account in the complexity of the distributed computations in this model. Nevertheless, in this work, the running time of local computations is at most polynomial in the size of the input. The running time of a distributed algorithm is defined to be the number of communication rounds.

A labeled graph is a pair $(G,x)$ where $G=(V,E)$ is a connected simple graph, and $x:V\to \{0,1\}^*$ is a  function assigning a bit-string, called a \emph{label}, to every node of~$G$.
When discussing a weighted $n$-node graph $G$, we assume $G=(V,E,w)$, where $w:E\to [1,n^c]$ for a fixed $c\geq 1$, so $w(e)$ can be encoded on $O(\log n)$ bits.
An identity-assignment to a graph $G$ is an assignment $\id:V\to [1,n^c]$, for some fixed $c\geq 1$, of distinct identities to the nodes.

A distributed decision algorithm is an algorithm in which every node outputs accept or reject. We say that such an algorithm accepts if and only if every node outputs accept.

Given a finite collection $\mathcal{G}$ of labeled graphs, we consider a Boolean predicate $\mathcal{P}$ on every labeled graph in~$\mathcal{G}$ (which may even depend on the identities assigned to the nodes). 
For instance, $\mbox{\sc aut}$ is the predicate on graphs stating that there exists a non-trivial automorphism of the graph (Recall that an automorphism of a graph is permutation $\sigma$ of the nodes, such that for all pair of nodes $u,v$, $(u,v)$ is an edge if and only if $(\sigma(u), \sigma(v))$ is an edge, and that it is non-trivial if the permutation is not the identity.)
Similarly, for any weighted graph with identity-assignment $\id$, the predicate $\mbox{\sc mst}$ on $(G,x,\id)$ states whether $x(v)=\id(v')$ for some $v'\in N[v]$\footnote{In a graph, $N(v)$ denotes the set of neighbors of node~$v$, and $N[v]=N(v)\cup\{v\}$.} for every $v\in V(G)$, and whether the collection of edges $\{\{v,x(v)\}, v\in V(G)\}$ forms a minimum-weight spanning tree of $G$. 

\begin{definition}

A proof-labeling scheme for a predicate $\mathcal{P}$ is a pair $(\mathbf{p},\mathbf{v})$, where
\begin{itemize}
	\item $\mathbf{p}$, called \emph{prover}, is an oracle that assigns a bit-string called a \emph{certificate} to every node of every labeled graph $(G,x)\in \mathcal{G}$, potentially using the identities assigned to the nodes, and
	\item $\mathbf{v}$, called \emph{verifier}, is a distributed decision algorithm such that,
	for every $(G,x)\in \mathcal{G}$, and for every identity assignment $\id$ to the nodes of~$G$,
	\[
	\left\{\begin{array}{lcl}
	(G,x,\id) \; \mbox{satisfies} \; \mathcal{P} & \Longrightarrow &  \mathbf{v}\circ \mathbf{p}(G,x,\id) = \mbox{accept};\\
	(G,x,\id) \; \mbox{does not satisfy} \; \mathcal{P} & \Longrightarrow &  \mbox{for every prover $\mathbf{p'}$}, \; \mathbf{v}\circ \mathbf{p'}(G,x,\id) = \mbox{reject};
	\end{array}\right.
	\]
\end{itemize}
here, $\mathbf{v}\circ \mathbf{p}$ is the output of the verifier $\mathbf{v}$ on the certificates assigned to the nodes by $\mathbf{p}$.
That is, if $(G,x,\id)$ satisfies $\mathcal{P}$, then, with the certificates assigned to the nodes by the prover $\mathbf{p}$, the verifier accepts at all nodes. Instead, if $(G,x,\id)$ does not satisfy $\mathcal{P}$, then, whatever certificates are assigned to the nodes, the verifier rejects in at least one node.
We emphasize that in this work, the verifier is always deterministic.
\end{definition}

The \emph{radius} of a proof-labeling scheme $(\mathbf{p},\mathbf{v})$  is defined as  the maximum number of rounds of the verifier $\mathbf{v}$ in the \LOCAL{} model~\cite{Peleg00}, over all identity-assignments to all the instances in $\mathcal{G}$, and all arbitrary certificates. It is denoted by $\radius(\mathbf{p},\mathbf{v})$.
Often in this paper, the phrase proof-labeling scheme is abbreviated to PLS, while a proof-labeling scheme of radius~$t\geq 1$ is abbreviated to $t$-PLS. Note that, in a $t$-PLS, one can assume, w.l.o.g., that the verification procedure, which is given $t$ as input to every node, proceeds at each node  in two phases: (1)~collecting all the data (i.e., labels and certificates) from nodes at distance at most~$t$, including the structure of the ball of radius~$t$ around that node, and (2)~processing all the information for producing a verdict, either accept or reject.
Note that, while the examples in this paper are of highly uniform graphs, and thus the structure of the $t$-balls might be known to the nodes in advance, our scaling mechanisms work for arbitrary graphs.

Given an instance $(G,x,\id)$ satisfying $\mathcal{P}$, we denote by $\mathbf{p}(G,x,\id,v)$ the certificate assigned by the prover $\mathbf{p}$ to node $v\in V$, and by $|\mathbf{p}(G,x,\id,v)|$ its size. We also let $|\mathbf{p}(G,x,\id)|=\max_{v\in V(G)} |\mathbf{p}(G,x,\id,v)|$. The \emph{certificate-size} of a proof-labeling scheme $(\mathbf{p},\mathbf{v})$ for $\mathcal{P}$ in $\mathcal{G}$, denoted $\size(\mathbf{p},\mathbf{v})$, is defined as the maximum of $|\mathbf{p}(G,x,\id)|$, taken over all instances $(G,x,\id)$ satisfying~$\mathcal{P}$, where $(G,x)\in \mathcal{G}$.  
In the following, we focus on the graph families $\mathcal{G}_n$ of connected simple graphs with $n$ nodes, $n\geq 1$. That is, the size of a proof-labeling scheme is systematically expressed as a function of the number~$n$ of nodes. For the sake of simplifying the presentation, the graph family $\mathcal{G}_n$ is omitted from the notations.

The minimum certificate size of a $t$-PLS for the predicate $\mathcal{P}$ on $n$-node labeled graphs is denoted by $\pls(\mathcal{P},t)$, that is,
\[
\pls(\mathcal{P},t)= \min_{\radius(\mathbf{p},\mathbf{v}) \leq t} \size(\mathbf{p},\mathbf{v}).
\]
We also denote by  $\pls(\mathcal{P})$ the size of a standard (radius-1) proof-labeling scheme for $\mathcal{P}$, that is, $\pls(\mathcal{P}) =\pls(\mathcal{P},1)$. For instance, it is known that $\pls(\mbox{\sc mst})=\Theta(\log^2n)$ bits~\cite{KormanK07,KormanKP10}, and that $\pls(\mbox{\sc aut})=\widetilde{\Omega}(n^2)$ bits~\cite{GoosS16}. More generally, for every decidable predicate $\mathcal{P}$, we have $\pls(\mathcal{P})=O(n^2+nk)$ bits~\cite{GoosS16} whenever the input labels are of $k$~bits, and $\pls(\mathcal{P},D)=0$ for graphs of diameter~$D$ because the verifier can gather all labels, and all edges at every node in $D$ rounds.

\begin{definition}
	Let $\mathcal{I} \subseteq \mathbb{N}^+$, and let $f:\mathcal{I}\to \mathbb{N}^+$. Let~$\mathcal{P}$ be a Boolean predicate on labeled graphs.
	A set $(\mathbf{p}_t,\mathbf{v}_t)_{t\in \mathcal{I}}$ of proof-labeling schemes for~$\mathcal{P}$, with respective radius~$t\geq 1$,  \emph{scales} with scaling factor~$f$ on~$\mathcal{I}$ if $\size(\mathbf{p}_t,\mathbf{v}_t)= O\big(\frac{1}{f(t)}\cdot \pls(\mathcal{P})\big)$ bits for every $t\in\mathcal{I}$.
	The set $(\mathbf{p}_t,\mathbf{v}_t)_{t\in \mathcal{I}}$ \emph{weakly scales} with scaling factor~$f$ on $\mathcal{I}$ if 
	 $\size(\mathbf{p}_t,\mathbf{v}_t)= O\big(\frac{\poly\log n}{f(t)}\cdot \pls(\mathcal{P})\big)$ 
	
	bits for every $t\in\mathcal{I}$.
\end{definition}

In the following, somewhat abusing terminology, we shall say that a proof-labeling scheme (weakly) scales while, formally, it should be a set of proof-labeling schemes that scales.

\begin{remark}
At  first glance, it may seem that no proof-labeling scheme can scale more than linearly, i.e., one may be tempted to claim that for every predicate $\mathcal{P}$ we have $\pls(\mathcal{P},t)=\Omega\left(\frac{1}{t} \cdot \pls(\mathcal{P})\right)$.
The rationale for such a claim is that, given a proof-labeling scheme $(\mathbf{p}_t,\mathbf{v}_t)$ for~$\mathcal{P}$, with radius~$t$ and  $\pls(\mathcal{P},t)$, one can construct a proof-labeling scheme $(\mathbf{p},\mathbf{v})$ for $\mathcal{P}$ with radius~$1$ as follows: the certificate of every node~$v$ is the collection of certificates assigned by~$\mathbf{p}_t$ to the nodes in the ball of radius~$t$ centered at~$v$; the verifier~$\mathbf{v}$ then simulates the execution of $\mathbf{v}_t$ on these certificates.
In paths or cycles, the certificates resulting from this construction are of size $O(t\cdot \pls(\mathcal{P},t))$, from which it follows that no proof-labeling scheme can scale more than linearly.
There are several flaws in this reasoning, which actually make it erroneous.

First, it might be the case that degree-2 graphs are not the worst case graphs for the predicate $\mathcal{P}$; 
for instance, if the predicate $\mathcal{P}$ requires large certificates only on degree-3 graphs, then the same rational cannot rule out much better scaling, as it only leads to
$\pls(\mathcal{P},t)=\Omega\left(\frac{1}{2^t} \cdot \pls(\mathcal{P})\right)$.
Second, in $t$ rounds of verification
every node learns not only the certificates of its $t$-neighborhood,
but also its structure,
which may contain valuable information for the verification;
this idea stands out when the lower bounds for $\pls(\mathcal{P})$ (without scaling) are established using labeled graphs of constant diameter, in which case there is no room for studying how proof-labeling schemes can scale.

The take away message is that establishing lower bounds of the type $\pls(\mathcal{P},t)=\Omega(\frac{1}{t} \cdot \pls(\mathcal{P}))$ for $t$ within some non-trivial interval requires specific proofs, which often depend on the given predicate~$\mathcal{P}$.
\end{remark}

\paragraph*{Communication Complexity.}
We use results from the well studied field of communication complexity~\cite{KN,Yao79}.
In the set-disjointness (\disj) problem on $k$ bits, each of
two players, Alice and Bob, is given a $k$-bit string, denoted $S_A$ and $S_B$, respectively, as input.
They aim at deciding whether there does not exist $i\in\set{1,\ldots,k}$
such that $S_A[i]=S_B[i]=1$.
We associate sets with characteristic vectors, i.e., the set includes all indexes $i$ for which $S[i]=1$. 
A standard abuse of notation leads to the following alternative definition of this problem.
The goal in set-disjointness is to deciding whether $S_A\cap S_B=\emptyset$.
The communication complexity of a given protocol solving \disj{}
is the number of bits Alice and Bob must communicate, in the worst case,
when using this protocol.
The communication complexity of \disj{} is the minimum
communication complexity of a protocol solving it.

In \emph{nondeterministic communication complexity},
each of the players receives,
in addition to its input, a binary string as
a hint, which may depend on both input strings, and the following holds.
For every TRUE-instance, there exists a hint such that the players output `TRUE'; and for every FALSE-instance, for all hints, the players output `FALSE'.
For example, a good hint for $\overline{\mbox{\disj}}$,
i.e., deciding whether there exists $i\in\set{1,\ldots,k}$
such that $S_A[i]=S_B[i]=1$, is the index $i$ itself. Indeed,
if one of the two players receives $i$ as a hint,
he or she can send it to the other player,
and they both check that $S_A[i]=1$ and $S_B[i]=1$.

The communication complexity of a nondeterministic protocol for \disj{}
is the number of bits the players exchange on two input strings that are disjoint,
in the worst case,
when they are given optimal nondeterministic strings.
The nondeterministic communication complexity of \disj{} is the minimum,
among all nondeterministic protocols for \disj{},
of the communication complexity of that protocol.
The following theorem holds.
The nondeterministic communication complexity of \disj{} is known to be $\Omega(k)$,
as a consequence of, e.g., Example~1.23 and Definition~2.3 in~\cite{KN}.

\begin{theorem}\label{thm:non-det-disj}
The nondeterministic communication complexity of \disj{} is $\Omega(k)$.
\end{theorem}

\section{All Proof-Labeling Schemes Scale Linearly in Trees}\label{sec:scaling-on-trees}

This section is entirely dedicated to the proof of one of our main results, stating that \emph{every} predicate on labeled trees has a proof that scales linearly. Further in the section, we also show how to extend this result to cycles and to grids, and, more generally, to multi-dimensional grids and toruses.

\subsection{Linear scaling for trees}%
\begin{theorem}\label{thm:tree-scaling}
	Let  $\mathcal{P}$ be a predicate on labeled trees. Then  $\pls(\mathcal{P},t)=O\left(\ceil*{\frac{\size(\mathcal{P})}{t}}\right)$.	
\end{theorem}


\begin{remark}
Here and in other places in the paper, we use ceilings in the asymptotics, because the size of the certificates needs to be at least one bit.
\end{remark}

The rest of this subsection is dedicated to the proof of Theorem~\ref{thm:tree-scaling}. So, let  $\mathcal{P}$ be a predicate on labeled trees, and let $(\mathbf{p},\mathbf{v})$ be a proof-labeling scheme  for $\mathcal{P}$ with $\size{(\mathbf{p},\mathbf{v})}=k$. First, note that we can restrict attention to trees with diameter $>t$. Indeed, predicates on labeled trees with diameter $\leq t$ are easy to verify since every node can gather the input of the entire tree in $t$ rounds. More precisely, if we have a scheme that works for trees with diameter~$>t$, then we can trivially design a scheme that applies to all trees, by adding a single bit to the certificates, indicating whether the tree is of diameter at most $t$ or not.

The certificate assignment in our scaling scheme is based on a specific decomposition of given tree~$T$.

\subparagraph{Tree decomposition.}

\begin{figure} 
	\centering
	\vspace{-4mm}\includegraphics[width=2.7in]{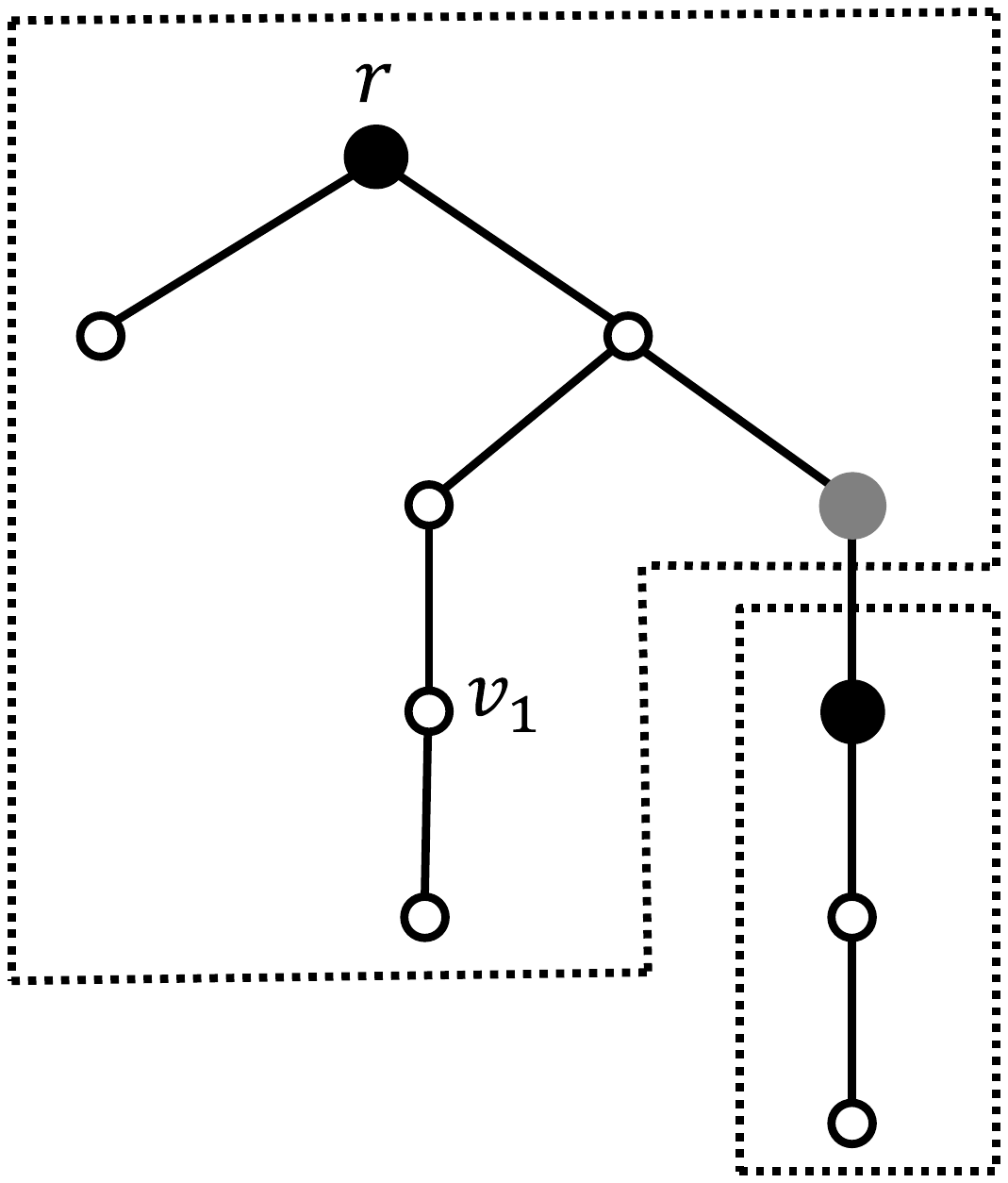}
	\caption{\it Tree decomposition for $t=6$ (i.e., $h=3$). Black nodes are \emph{border} nodes, the gray node is an \emph{extra-border} node, and white nodes are \emph{standard} nodes. 
	Note that $v_1$ is not a border node since the depth of its subtree is $1<h-1$. The domains are marked with dashed frames.}
	\label{fig-tree-scale}
\end{figure}

A key tool in the proof is a decomposition of the tree. It is defined the following way.

\begin{definition}\label{def:tree-decomposition}
Let $T$ be a tree of diameter $>t$, and let
\[
h=\lfloor t/2 \rfloor.
\]
The tree $T$ is rooted at some node~$r$.  A node $u$ such that
\[
\dist_T(r,u) \equiv 0\pmod h,
\]
and $u$ possesses a subtree of depth at least $h-1$ is called a \emph{border} node. Similarly, a node $u$ such that
\[
\dist_T(r,u) \equiv -1 \pmod h,
\]
and $u$ possesses a subtree of depth at least $h-1$ is called an \emph{extra-border} node. A node that is a border or an extra-border node is called a \emph{special} node. All other nodes are  \emph{standard} nodes. For every border node $v$, we define the \emph{domain} of $v$ as the set of nodes in the subtree rooted at $v$ but not in subtrees rooted at border nodes that are descendants of $v$.
\end{definition}

Figure~\ref{fig-tree-scale} illustrates of the tree decomposition, and can also serve as intuition for the following lemma.

\begin{lemma}\label{lem:marking} ~
	\vspace{-1ex}
	\begin{enumerate}
		\setlength\itemsep{1pt}
		\item The domains form a partition of the nodes in the tree $T$.
		\item Every  domain forms a tree rooted at a border node, with depth in the range $[h-1,2h-1]$.
		\item Two adjacent nodes of $T$ are in different domains if and only if they are both special.
	\end{enumerate}
\end{lemma}

\begin{proof}
	Let us first prove the first item. On the one hand, every node belongs to a domain. This is because every node has at least one border ancestor, since the root is a border node. Indeed, the root $r$ has depth $0$, and the diameter of $T$ is at least $t+1$, which implies that $r$ necessarily possesses a subtree of depth at least $h-1$. On the other hand, every node $u$ belongs to a unique domain. This is simply because the closest border ancestor of $u$ is uniquely defined.
	
	To establish second item, we consider the domain of a border node $u$. Note that, for any node $v$ in the domain of $u$, $v$ is a descendent of $u$ in $T$, and all the nodes in the shortest path between $u$ and $v$ are also in the domain of $u$. Thus the domain of $u$ is indeed a tree rooted at $u$. The depth of a domain is at least $h-1$. Indeed, if the subtree rooted at~$u$ has depth $<h-1$, then, by definition, $u$ is not special. It remains to show that the depth of a domain is at most $2h-1$. Let us assume for the purpose of contradiction that  the depth~$d$ of the domain of $u$ satisfies $d>2h-1$. Then there exists a path of length at least $2h$ starting at $u$, and going downward the tree for reaching a leaf $v$ of this domain. Then let us consider the node $u'$ of that path which is at distance $h$ from~$u$. Node $u'$ is not a border node, since otherwise $v$ would not be in domain of $u$ but in the domain of $u'$. However, node $u'$ is a border node as its depth is $0$ modulo~$k$, and it has a subtree of depth at least $d-h>h-1$. This contradiction completes the proof of item~2.
	
Finally, for establishing the third item, let us consider an edge $\{u,v\}$ in the tree $T$, and let us assume, w.l.o.g., that  $u$ is the parent of $v$ in $T$. By construction, there can be  three cases only. If neither of the two nodes $u$ and $v$ is a border node, then both belong to the same domain, as they have the same closest border ancestor. If $u$ is a border node, and $v$ is a standard node, then $v$ is in the domain of $u$. Finally, if $v$ is a border node, then necessarily~$u$ is an extra-border node, in which case $u$ and $v$ do not belong to the same domain since $v$ is in its own domain, while $u$ cannot belong to the domain of $v$. Therefore, Item~3 holds in these three cases.
\end{proof}

\subparagraph{Verifying the decomposition.}

The certificates of the radius-$t$ proof-labeling scheme contain a 2-bit field indicating to each node whether it is a root, border, extra-border, or standard node. Let us show that this part of the certificate can be verified in $t$ rounds. The prover orients the edges of the tree towards the root $r$.

Such an orientation can be given to the edges of a tree by assigning to each node its distance to the root, modulo 3. The assumption is that every edge is oriented from $d$ to $d-1 \mod 3$. Hence, a node with distance $d$, is either the root, and then and all of its neighbors are at distance $d+1 \mod 3$, or it has exactly one neighbor with distance $d-1 \mod 3$ which is its parent in the tree. These distances can obviously be checked locally, in just one round. So, in the remaining of the proof, we assume that the nodes are given this orientation upward the tree. The following lemma proves that the decomposition into border, extra-border, and standard nodes can be checked in $t$~rounds.

\begin{lemma}\label{lem:check-marking}
	Given a set of nodes marked as border, extra-border, or standard in an oriented tree, there is a verification protocol that checks whether the marking corresponds to a tree decomposition (as in Definition~\ref{def:tree-decomposition}), in $2h<t$ rounds.
\end{lemma}

\begin{proof}
	The checking procedure proceeds as follows. The root $r$ checks that it is a border node. Every border node checks that its subtree truncated at depth $2h$ fits with the decomposition. That is, it checks that: (1)~no nodes in its subtree are border nodes except nodes at depth $h$ and $2h$, (2)~no nodes  in its subtree are extra-border nodes except nodes at depth $h-1$ and $2h-1$, and  (3)~the nodes  in its subtree that are special at depth $h-1$ and $h$ do have a subtree of depth at least $h-1$. By construction, this procedure accepts any marking which is correct with respect to the decomposition rule.
	
	Conversely, let us suppose that the algorithm accepts a marking of the nodes. We prove that this marking is necessarily correct. We proceed by induction on the depth of the nodes. At the root, the verifier checks that $r$ is special as a border node, and it checks that the domain of $r$ is correctly marked. In particular, it checks that the nodes of depth $h$ that are not in its domain are properly marked as border nodes. So, the base of the induction holds. Now, assume that, for~$\alpha\geq 0$, all the domains whose border nodes stand at depth at most $\alpha h$ are properly marked. The fact that the border nodes at depth $\alpha h$ accept implies that all the nodes at depth $(\alpha+1) h$ that are not in the domain of the border nodes at depth $\alpha h$ are properly marked. These nodes verify their own domains, as well as all the domains down to depths $(\alpha+1) h$, are all correct. Since none of these nodes reject, it follows that all the domains whose border nodes stand at depth at most $(\alpha+1) h$ are properly marked. This completes the induction step, and hence the proof of the lemma.
\end{proof}

We are now ready to describe the radius-$t$ proof-labeling scheme for an arbitrary predicate. From Lemma~\ref{lem:check-marking} (and the paragraph above it that defines the decomposition), we can assume that the nodes are correctly marked as root, border, extra-border, and standard, with a consistent orientation of the edges towards the root.

\subparagraph{The radius-$t$ proof-labeling scheme. }

We are  considering the arbitrarily given predicate~$\mathcal{P}$ on labeled trees, with its proof-labeling scheme $(\mathbf{p},\mathbf{v})$ using certificates of size~$k$ bits. Before reducing the size of the certificates to $O(k/t)$ by communicating at radius~$t$, we describe a proof-labeling scheme at radius~$t$ which still uses large certificates, of size $O(k)$, but stored at a few nodes only, with all other nodes storing no certificates.

\begin{lemma}\label{lem:labels-on-marked-nodes}
	There exists a radius-$t$ proof-labeling scheme for $\mathcal{P}$, in which the prover assigns certificates to special nodes only, and these certificates have size $O(k)$.
\end{lemma}

\begin{proof}
	On legally labeled trees, the prover provides every special node (i.e., every border or extra-border node) with the same certificate as the one provided by~$\mathbf{p}$. All other nodes are provided with no certificates.
	
	On arbitrary labeled trees, the verifier is active  at border nodes only, and all non-border nodes systematically accept (in zero rounds). At a border node $v$, the verifier first gathers all information at radius $2h$. This includes all the labels of the nodes in its domain, and of the nodes that are neighbors to a node in its domain. Then $v$ checks whether there exists an assignment of $k$-bit certificates to the standard nodes in its domain that results in $\mathbf{v}$ accepting at every node in its domain. If this is the case, then $v$ accepts, else it rejects. There is a subtle point worth to be mentioned here. The value of~$k$ may actually depend on~$n$, which is not necessarily known to the nodes. Nevertheless, this can be easily fixed as follows. The $t$-PLS prover is required to provide all nodes with certificates of the same size (the fact that all certificates have identical size can trivially be checked in just one round). Then $k$~is simply inferred from the certificate size in the $t$-PLS, by multiplying this size by~$t$, whose value is, as specified in Section~\ref{sec:model-and-notations}, given as input to each node.
	
	Note that, as every border node $v$ has a complete view of its whole domain, and of the nodes at distance~1 from its domain, $v$ considers all the nodes that are used by $\mathbf{v}$ executed at the nodes of its domain. Also note that the execution of $\mathbf{v}$ at nodes in the domain of $v$ concerns only nodes that are either in the domain of $v$, or are special. This follows from the third item in Lemma \ref{lem:marking}. Thus no two border nodes will simulate the assignment of certificates to the same node.
	
	We now prove that, in an oriented marked tree, this scheme is correct.
	
	-- Assume first that the labeled tree satisfies the predicate  $\mathcal{P}$. Giving to the special nodes the certificates as assigned by $\mathbf{p}$, all the border nodes will be able to find a proper assignment of the certificates for the standard nodes in their domain so that $\mathbf{v}$ accepts at all these nodes, since, as the labeled tree satisfies the predicate  $\mathcal{P}$, there must exists at least one. This leads every node to accept.
	
	-- Suppose now that every border node accepts. It follows that, for every border node, there is an assignment of certificates to the nodes in its domain such that $\mathbf{v}$ accepts these certificates at every node. The union of these partial assignments of certificates defines a certificate assignment to the whole tree that is well-defined according to the first item of Lemma \ref{lem:marking}. At every node, $\mathbf{v}$ accepts since it has the same view as in the simulation performed by the border nodes in their respective domains.  Thus $\mathbf{v}$ accepts at every node of $T$, and therefore it follows that the labeled tree satisfies  $\mathcal{P}$.
\end{proof}

Lemma \ref{lem:labels-on-marked-nodes} basically states that there is a radius-$t$ proof-labeling scheme in which the prover can give certificates  to special nodes only. We now show how to spread out the certificates of the border and extra-border nodes to obtain smaller certificates. The following lemma is the main tool for doing so. As this lemma is also used further in the paper, we provide a generalized version of its statement, and we later show how to adapt it to the setting of the current proof.

We say that a local algorithm~$\mathcal{A}$ \emph{recovers} an assignment of certificates provided by some prover $\mathbf{q}$ from an assignment of certificates provided by another prover $\mathbf{q}'$ if, given the certificates assigned by~$\mathbf{q}'$ as input to the nodes, $\mathcal{A}$ allows every node to output its certificate such as assigned by $\mathbf{q}$. We define a \emph{special} prover as a prover which assigns certificates only to the special nodes, all other nodes being given empty certificates.

\begin{lemma}\label{lem:spreading}
	For every prover $\mathbf{q}$, there exists a $\LOCAL$ algorithm $\mathcal{A}$ and a prover $\mathbf{q}'$ satisfying the following. For every $s \geq 1$, for every oriented marked tree $T$ of depth at least $s$, and for every assignment of $b$-bit labels provided by  $\mathbf{q}$ to the nodes of $T$, there exists assignment of $O(\ceil*{b/s})$-bit labels provided by  $\mathbf{q}'$ to the nodes of $T$ such that $\mathcal{A}$ recovers $\mathbf{q}$ from $\mathbf{q}'$ in $s$ rounds.
\end{lemma}

\begin{proof}
	We first describe the prover $\mathbf{q}'$. For each special node $v$, let us partition the certificate $\mathbf{q}(v)$ assigned to node $v$ by the special prover $\mathbf{q}$ into $s$ parts of size at most $\ceil*{b/s}$. Then one picks an arbitrary path starting from $v$, of length $s-1$, going downward the tree. Note that such a path exists by definition of border and extra-border nodes. For every $i\in\{0,\dots,s-1\}$, the $i$th part of the certificate  $\mathbf{q}(v)$ is assigned to the $i$th node of that path as its certificate in $\mathbf{q}'$. As such, every node is given at most two parts of the initial certificates, as the paths starting at each of the border nodes are non intersecting, and the paths starting at each of the extra-border nodes are non intersecting as well.  To recover the original certificates, for every special node $v$, the algorithm $\mathcal{A}$ simply inspects the tree at distance $s-1$ downward, for gathering all the parts of the initial certificate $\mathbf{q}(v)$ of $v$. Then $v$ concatenates these parts, and $v$ outputs the resulting certificate. All other nodes output a certificate formed by the empty string.
\end{proof}

We have now all the ingredients to prove Theorem  \ref{thm:tree-scaling}.

\subparagraph{Proof of Theorem  \ref{thm:tree-scaling}.}

In the radius-$t$ proof-labeling scheme, the prover chooses a root and an orientation of the tree $T$, and provides every node with a counter modulo~3 in its certificate allowing the nodes to check the consistency of the orientation. Then the prover constructs a tree decomposition of the rooted tree, and provides every node with its type (root, border, extra-border, or standard) in its certificates. Applying Lemmas~\ref{lem:labels-on-marked-nodes} and~\ref{lem:spreading}, the prover spreads the certificates assigned to the special nodes by $\mathbf{p}$. Every node will get at most two parts, because only the paths associated to a border node and to its parent (an extra-border node) can intersect. Overall, the certificates have size $O(k/h)=O(k/t)$. The verifier checks the orientation and the marking, then recovers the certificates of the special nodes, as in Lemma~\ref{lem:spreading}, and performs the simulation as in Lemma~\ref{lem:labels-on-marked-nodes}. This verification can be done with a view of radius $t\leq 2h$, yielding the desired radius-$t$ proof labeling scheme. \qed

\subsection{Linear scaling in cycles and grids}

For the proof techniques of Theorem \ref{thm:tree-scaling} to apply to other graphs, we need to compute a partition of the nodes into the two categories, special and standard, satisfying three main properties. First, the partition should split the graph into regions formed by standard nodes, separated by special nodes. Second, each region should have a diameter small enough for allowing special nodes at the border of the region to simulate the standard nodes in that region, as in Lemma~\ref{lem:labels-on-marked-nodes}. Third,  the regions should have a diameter large enough to allow efficient spreading of certificates assigned to special nodes over the standard nodes, as in Lemma~\ref{lem:spreading}. For any graph family in which one can define such a decomposition, an analogue of Theorem~\ref{thm:tree-scaling} holds. We show that this is the case for cycles and grids.

\begin{corollary}\label{cor:cycle-and-grid-scaling}
	Let  $\mathcal{P}$ be a predicate on labeled cycles.
	If there
	exists a (radius-1) proof-labeling scheme $(\mathbf{p},\mathbf{v})$ for $\mathcal{P}$ with $\size{(\mathbf{p},\mathbf{v})}=k$, then $\pls(\mathcal{P},t)=O\left(\ceil*{\frac{k}{t}}\right)$. The same holds for predicates on  2-dimensional labeled grids.
\end{corollary}


\begin{figure} 
	\centering
	\vspace{-4mm}\includegraphics[width=3.5in]{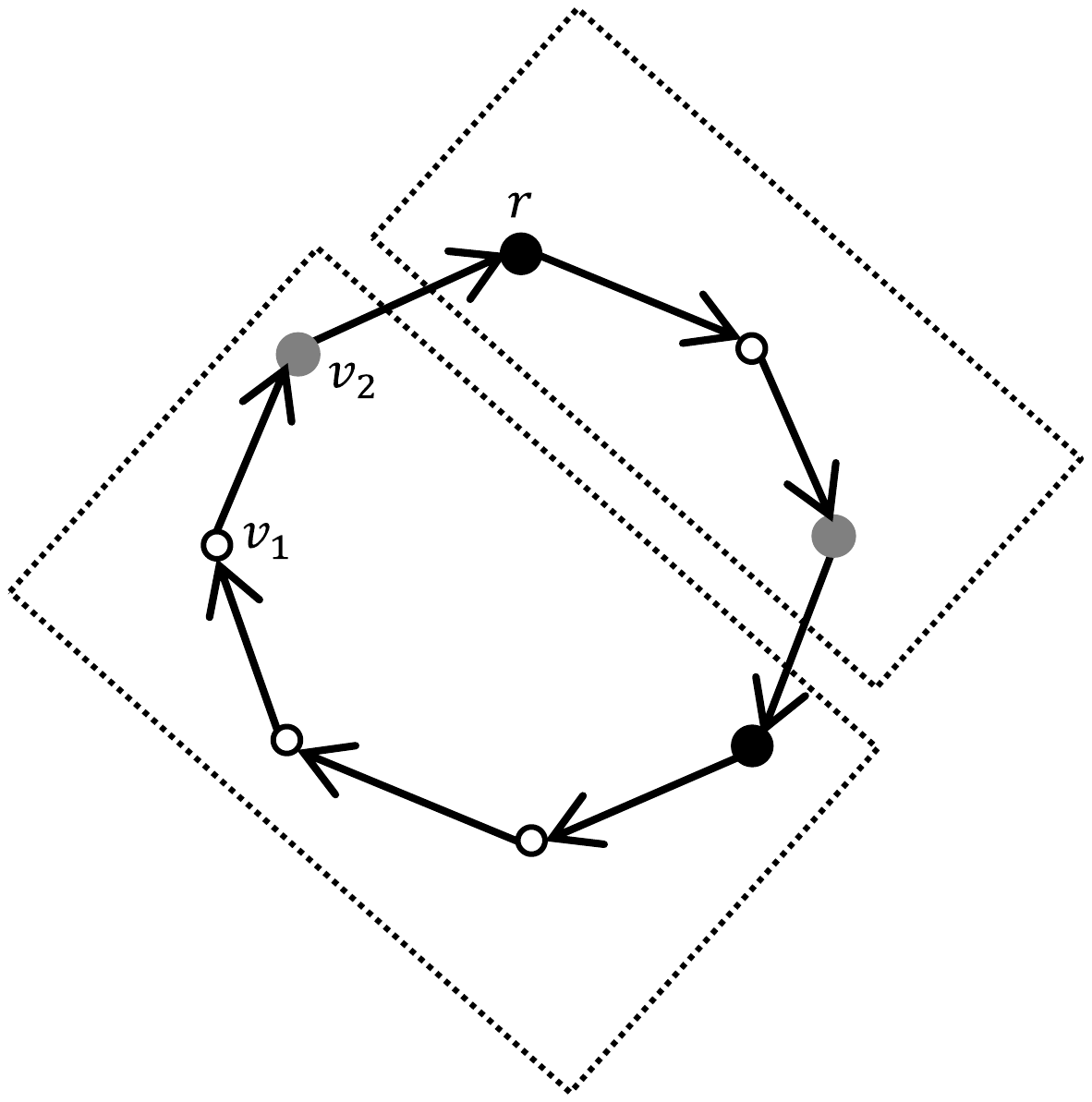}
	\caption{\it Cycle decomposition for $t=6$  (i.e., $h=3$). Black nodes are \emph{border} nodes,  gray nodes are  \emph{extra-border} nodes, and white nodes are \emph{standard} nodes. Note that $v_1$ is not a border node since the number of nodes between $v_1$ and $r$ in the oriented cycle is $1<h-1$, and $v_2$ is  an extra-border node since this is the farthest node from $r$ in the oriented cycle. The domains are marked with dashed frames.}
	\label{fig-cycle-scale}
\end{figure}

\begin{figure} 
	\centering
	\vspace{-4mm}\includegraphics[width=4in]{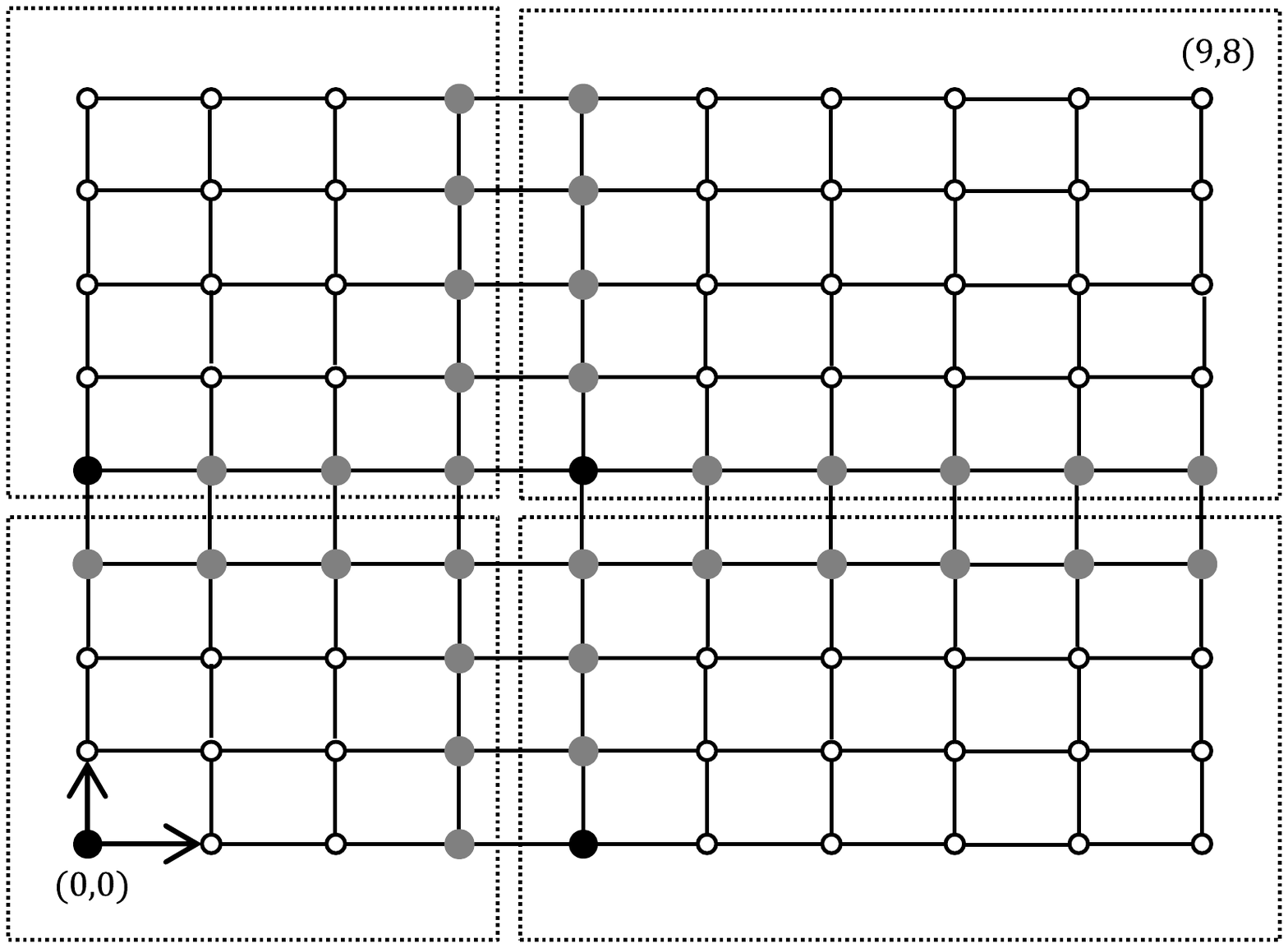}
	\caption{\it Grid decomposition for $t=16$ (i.e., $h=4$). Black nodes are \emph{border} nodes,  gray nodes are \emph{extra-border} nodes, and white nodes are \emph{standard} nodes.  The domains are marked with dashed frames.}
	\label{fig-grid-scale}
\end{figure}

See Figure~\ref{fig-cycle-scale} for an example of cycle decomposition. The idea is to orient the cycle edges in a consistent way, starting from an arbitrary node $r$. The partition to domains is exactly as in a tree,  while ignoring the last oriented edge that closes the cycle with $r$ (e.g., the edge $(v_2,r)$ in Figure~\ref{fig-cycle-scale}). The only difference from a tree is that the farthest node from $r$ in the oriented cycle ($v_2$ in Figure~\ref{fig-cycle-scale}) must be an extra-border node in order to make a proper partition in which  there are two spacial nodes on every boundary.
Also, see Figure~\ref{fig-grid-scale} for an example of grid decomposition.

\begin{proof}[Proof of Corollary~\ref{cor:cycle-and-grid-scaling}.]
	We just explain how the proof of Theorem~\ref{thm:tree-scaling}  can be adapted to apply for cycles and grids.
	
	For every labeled cycle, the prover picks an arbitrary node $r$, which will play the same role as the root chosen in the proof of Theorem~\ref{thm:tree-scaling}, and an orientation of the cycle pointing toward $r$. The chosen node is called \emph{leader}. Let
	$
	h=\lfloor t/2 \rfloor.
	$
	The prover marks as border node every node $u$ such that
	$
	\dist(r,u) \equiv 0\pmod h,
	$
	where the distance is taken in the oriented cycle. Similarly, the prover marks as extra-border node every node $u$ such that
	$
	\dist(r,u) \equiv -1\pmod h.
	$
	Note that the border and extra-border nodes that are the furthest away from the leader by the orientation may actually be close to the leader in the undirected cycle. As this may cause difficulties for spreading the certificates, the prover does not mark them, and keep these nodes standard. The domain of a border node is defined in a way similar to trees. The leader, the orientation, and the marking can be checked in a way similar to the proof of Lemma~\ref{lem:check-marking}. In particular, observe that the size of each domain is  at most~$t$. The marking separates the graph into independent domains that can be simulated in parallel as in Lemma~\ref{lem:labels-on-marked-nodes}. The diameter of each domain is at least $t/2$ which allows one to do the spreading as in Lemma~\ref{lem:spreading}, resulting in certificates of size $O(k/t)$.
	
	For every labeled grid, the prover provides the edges with north-south and east-west orientations, using two counters modulo~3. In a grid $p\times q$ with $n=pq$ nodes, this orientation induces a coordinate system with edges directed from $(0,0)$, defined as the south-west corner, to $(p,q)$, defined as the north-east corner. The leader is the node at position $(0,0)$. Let $h=\lfloor t/4\rfloor$. The partition of the nodes is as follows. Every node with coordinate $(x,y)$ where both $x$ and $y$ are $0$ modulo $h$ are the border nodes. Non-border nodes with coordinate $(x,y)$ where $x$ or $y$ equals $0$ or $-1$  modulo $h$ are extra-border nodes. Now, as for cycles, we slightly modify that decomposition for avoiding domains with too small diameter. Specifically, north-most border and extra-border nodes, and the east-most border and extra-border nodes are turned back to standard nodes. The domain of a border node is composed of all nodes with larger $x$-coordinate and larger $y$-coordinates for which there are no closer border node in the oriented grid. Using the same technique as in the proof of Lemma~\ref{lem:check-marking}, this partition of the nodes can be checked locally. The simulation as performed in the proof of Lemma~\ref{lem:labels-on-marked-nodes} can be performed similarly using that decomposition, because the domains have diameter at most $t$, and are well separated by special nodes. Finally, the spreading of the certificates as in Lemma~\ref{lem:spreading} can be done in the following way. For every special node $(x,y)$ where $x$ equals $0$ or $-1$ modulo~$h$, the certificate is spread over the $h-1$ nodes to the east. For every special node $(x,y)$ where $y$ equals $0$ or $-1$ modulo~$h$, the certificate is spread over the $h-1$ nodes to the north. Note that there is always enough space to the east or to the north of the special nodes as we have removed the special nodes that could be too close to the east and north borders. Also note that some node have their certificates spread in two directions, but this does not cause problem as it just increases the size of the certificates by a constant factor.
\end{proof}

By the same techniques, Corollary~\ref{cor:cycle-and-grid-scaling} can be generalized to toroidal 2-dimensional labeled grids, as well as to $d$-dimensional labeled grids and toruses, for every constant $d\geq 2$.

\section{Growth-Dependent Scaling of Uniform Proof-Labeling Schemes}
\label{sec:optimal-uniform}

In Section~\ref{sec:scaling-on-trees}, we presented a general tradeoff for every predicate, on \emph{specific graph structures}. In this section, we give a general tradeoff on \emph{every graph structure}, for specific type of schemes.
We say that a graph $G=(V,E)$ has \emph{growth}~$b=b(t)$ if, for every $v\in V$ and $t\in[1,D]$, we have that $|B_G(v,t)|\geq b(t)$. 
We say that a proof-labeling scheme is \emph{uniform} if the same certificate is assigned to all nodes by the prover.

\begin{theorem}\label{theo:universal}
Let $\mathcal{G}$ be a family of graphs with growth factor  $b$,
let  $\mathcal{P}$ be a predicate on labeled graphs in $\mathcal{G}$, and let us assume that there exists a uniform 1-PLS $(\mathbf{p},\mathbf{v})$ for $\mathcal{P}$ with $\size{(\mathbf{p},\mathbf{v})}=k$.

	Then, $\pls(\mathcal{P},t)=O\left(\ceil*{\frac{k}{b(t)}}\log n\log k\right)$ for every $t$ for which $b(t)  \geq \log n $.%

\end{theorem}

The rounding up operator $\ceil*{\cdot}$ is used to indicate that the scaling cannot reduce the logarithmic factors outside it, even in the case $\frac{k}{b(t)}<1$.

In the proof of this theorem, we use the probabilistic method to prove the existence of the required certificates. However, we emphasize that the scheme is deterministic.
We will use the classic Chernoff bounds (see, e.g.,~\cite[Chapter 4]{MitzenmacherU05}).

\begin{lemma}[Chernoff bounds]\label{lem:chernoff}
Suppose $Z_1,\dots, Z_m$ are independent random variables taking values in $\{0, 1\}$, and let $Z=\sum_{i=1}^mZ_i$. For every $0\leq \delta\leq 1$, we have $\Pr[Z\leq (1-\delta)\mathbb{E}Z]\leq e^{-\frac12 \delta^2\mathbb{E}Z}$, and $\Pr[Z\geq (1+\delta)\mathbb{E}Z]\leq e^{-\frac13 \delta^2\mathbb{E}Z}$.
\end{lemma}

\begin{proof}[Proof of Theorem~\ref{theo:universal}]
	
	Let $(\mathbf{p},\mathbf{v})$ be a uniform 1-PLS for $\mathcal{P}$ with $\size{(\mathbf{p},\mathbf{v})}=k$. Let $s=(s_1,\dots,s_k)$, where $s_i\in \{0,1\}$ for every $i=1,\dots,k$, be the $k$-bit certificate assigned to every node of $G$ by the prover $\mathbf{p}$.
	Let $t\geq 1$ be such that $k\geq b(t)\geq c \log n$ for a constant $c$ large enough. 
	For every node $v\in V$, set the certificate of $v$, denoted $s^{(v)}$, as follows: 
	for every $i=1,\dots,k$, $v$~stores the pair $(i,s_i)$ in~$s^{(v)}$ with probability $\frac{c \log n}{b(t)}$. 
	
	\begin{itemize}
		\item On the one hand, for every $v\in V$, let $X_v$ be the random variable equal to the number of pairs stored in~$s^{(v)}$. By a Chernoff bound (Lemma~\ref{lem:chernoff}), we have $\Pr[X_v\geq \frac{2c\, k \log n}{b(t)}]\leq e^{\frac{c\,k\log n}{3\,b(t)}} = n^{-\frac{c\,k}{3\,b(t)}}$. Therefore, by union bound,  the probability that a node~$v$ stores more than $\frac{2c\, k \log n}{b(t)}$ pairs $(i,s_i)$ is at most $n^{1-\frac{c\,k}{3\,b(t)}}$, which is less than $\frac12$ for $c$~large enough. 
		
		\item On the other hand, for every $v\in V$, and every $i=1,\dots,k$, let $Y_{v,i}$ be the number of occurrences of the pair $(i,s_i)$ in the ball of radius~$t$ centered at~$v$. By a Chernoff bound (Lemma~\ref{lem:chernoff}), we have $\Pr[Y_{v,i}\leq \frac{1}{2} c \log n]\leq e^{-\frac{c\log n}{8}}=n^{-c/8}$. Therefore, by union bound, the probability that there exists a node $v\in V$, and an index $i\in\{1,\dots,k\}$ such that none of the nodes in the ball of radius~$t$ centered at~$v$ store the pair $(i,s_i)$ is at most $kn^{1-c/8}$, which is less than $\frac12$ for $c$~large enough.
	\end{itemize}
	
	It follows that, for $c$ large enough, the probability that each node stores $O(k\log n/b(t))$ pairs $(i,s_i)$, and every pair $(i,s_i)$ is stored in at least one node of each ball of radius~$t$, is positive.

	Note that each pair can be represented by $O(\log k)$ bits.

	Therefore, there is a way for a prover to distribute the pairs $(i,s_i)$, $i=1,\dots,k$ to the nodes, such that 
	(1)~each node stores $O(k\log n\log k/b(t))$ bits, and 
	(2)~every pair $(i,s_i)$ appears at least once in every $t$-neighborhood of each node. 

	In the verification procedure, each node~$v$ first collects all pairs $(i,s_i)$ in its $t$-neighborhood, then recovers~$s$, and then runs the verifier of the original (radius-1) proof-labeling scheme.  
	
	Finally, we emphasize that we only use probabilistic arguments as a way to prove the existence of certificate assignment, 
	but the resulting proof-labeling scheme is deterministic and its correctness is not probabilistic.
\end{proof}


Theorem~\ref{theo:universal} finds direct application to the \emph{universal} proof-labeling scheme~\cite{KormanKP10,GoosS16}, which uses $O(n^2+kn)$ bits in $n$-node graphs labeled with $k$-bit labels. The certificate of each node consists of the $n\times n$ adjacency matrix of the graph, an array of $n$ entries each equals to the $k$-bit label at the corresponding node, and an array of $n$ entries listing the identities of the $n$~nodes. It was proved in~\cite{OstrovskyPR17} that the universal proof-labeling scheme can be scaled by a factor~$t$. Theorem~\ref{theo:universal} significantly improves that result, by showing that the universal proof-labeling scheme can actually be scaled by a factor~$b(t)$, which can be exponential in~$t$.

\begin{corollary}\label{cor:universal}
	For every predicate $\mathcal{P}$ on labeled graphs, there is a proof-labeling scheme for $\mathcal{P}$ as follows. For every graph $G$ with growth $b(t)$, let $t_0=\min\{t \mid b(t)  \geq \log n \}$. 

	Then, for every $t\ge t_0$  we have $\pls(\mathcal{P},t)=O\left(\frac{n^2+kn}{b(t)}\log n(\log n+\log k)\right)$. 

\end{corollary}

Theorem~\ref{theo:universal} is also applicable to proof-labeling schemes where the certificates have the same sub-certificate assigned to all nodes; in this case, the size of this common sub-certificate can be drastically reduced by using a $t$-round verification procedure. This is particularly interesting when the size of the common sub-certificate is large compared to the size of the rest of the certificates. An example of such a scheme is in essence the one described in~\cite[Corollary 2.4]{KormanKP10} for $\mbox{\sc iso}_k$. Given a parameter $k\in \Omega(\log n)$, let $\mbox{\sc iso}_k$ be the predicate on graph stating that there exist two vertex-disjoint isomorphic induced subgraphs of size~$k$ in the given graph.

\begin{corollary}
	For every $k\in\left[1,\frac{n}{2}\right]$, we have $\pls(\mbox{\sc iso}_k)=O(k^2+k\log n)$ and $\pls(\mbox{\sc iso}_k)=\Omega(k^2)$.
	In addition, for every $t>1$, we have 
	\[\pls(\mbox{\sc iso}_k,t)=O\left(\ceil*{\frac{k^2}{b(t)}}\log^2n\log\log n\right).
	\]%

\end{corollary} 

\begin{proof}
	We first sketch a 1-PLS. Every node is given as certificate the $k\times k$ adjacency matrix of the two isomorphic subgraphs, along with the corresponding IDs of the nodes in the two subgraphs.
	The certificates also provide each node with the ID of an arbitrary node in each subgraph, that we call the leaders. 
	In addition, the nodes are given certificates that correspond to two spanning trees rooted at the two leaders. 

	All this requires $O(k^2+k\log n)$-bit certificates.

	The verification procedure works as follows. 
	Every node first checks that the spanning trees structures are correct. Then the roots of the spanning trees check that they are marked as leader. Finally every node whose ID appears in one of the two adjacency matrices checks that its actual neighborhood corresponds to what it should be according to the given adjacency matrices. By construction, this is a valid 1-PLS, using certificates on $O(k^2+k\log n)$ bits. A simple adaptation of the proof of~\cite[Theorem 6.1]{GoosS16} enables one to prove that $\Omega(k^2)$ bits are needed. 
	
	For the case of $t$-PLS, a direct application of Theorem~\ref{theo:universal} to the part of the certificate that is common to all nodes gives a $t$-PLS for $\mbox{\sc iso}_k$ with $O\left(\ceil*{\frac{k^2+k\log n}{b(t)}}\log n\log(k^2+k\log n)\right)$-bit certificates.
	By analyzing the cases $k<\log n$ and $k\geq\log n$ separately, the corollary follows.%

\end{proof}

\section{Certifying Distance-Related Predicates}
\label{sec:distance-related}

In this section, we study two problems related to distances: diameter and spanners. 
For any labeled (weighted) graph $(G,x)$, the predicate $\mbox{\sc diam}$ on $(G,x)$ states whether, for every $v\in V(G)$,  $x(v)$ is equal to the (weighted) diameter of~$G$.
The next theorem states that there is a proof-labeling scheme for $\mbox{\sc diam}$ that weakly scales linearly, and that better scaling on all graphs is not possible. Theorem~\ref{thm: spanners scaling} states similar results for spanners.
It is unknown whether better scaling, and specifically scaling by factor $b(t)$, is possible for these problems on graphs with high growth.

\begin{theorem}
	\label{theo:scaling for diameter}
	There exists $c>0$ such that for every $t\in [c\log n,n/\log n]$, $\pls(\mbox{\sc diam},t)=O\left(\frac{n}{t}\log^2 n\right)$.
	Moreover, for every $t\in [1,n/\log n]$,
	$\pls(\mbox{\sc diam},t)=\Omega\left(\frac{n}{t\log n}\right)$.

\end{theorem}


The theorem follows from the next lemmas.

\begin{lemma}
	\label{lemma: pls diam scales linearly}
	There exists a constant $c$, such that for every $t\in [c\log n,n]$, $\pls(\mbox{\sc diam},t)=O\left(\frac{n}{t}\log^2 n\right)$.
\end{lemma}

\begin{proof}
	A proof-labeling scheme for diameter with optimal certificate size $\Theta(n \log n)$ bits has been designed in~\cite{Censor-HillelPP18}. We simply use this scheme for certifying that, for every node~$v$, the diameter of the graph is at least~$x(v)$. Indeed,~\cite{Censor-HillelPP18} uses only $O(\log n)$-bit certificates to certify the existence of a pair of nodes at mutual distance at least~$x(v)$ in the graph. The rest of the proof is dedicated to certifying that no pairs of nodes are at distance more than $x(v)$ in the graph, i.e., ``diameter $\leq x(v)$''. Namely, we show how the scheme in~\cite{Censor-HillelPP18} scales with the radius of verification. For this purpose, let us briefly recall this scheme. Each node $v$ of a graph $G=(V,E)$ is provided with a certificate $D_v$ consisting of a table with $n$~entries storing the ID of every node in $G$, and the distance to these nodes. (Every certificate is therefore on $O(n\log n)$ bits). Somewhat abusing notations, let us denote by $D_v(u)$ the distance to node $u$, as stored in table $D_v$. The verification proceeds by, first, having each node checking that it stores the same set of IDs as the ones stored by its neighbors, and that its own ID appears in its table. Second, each node checks that the distances in its certificate vary as expected. That is, each node $v$ checks that: (1)~$D_v(v)=0$, (2)~for every node $u$ and every neighbor~$v'$, $D_v(u)-w(\{v,v'\})\leq D_{v'}(u) \leq D_v(u)+w(\{v,v'\})$, and (3)~there exists a neighbor $v'$ such that $D_{v'}(u) = D_v(u)-w(\{v,v'\})$. Finally, every node $v$ checks that $D_v(u)\leq x(v)$ for every node $u$. This verification process is correct, as shown in~\cite{Censor-HillelPP18}. 
	
	Now, let $t \geq c \log n$ for a constant $c>0$ large enough, and let us construct a proof-labeling scheme for ``diameter $\leq x(v)$'', with radius~$t$. The idea is that each node $v$ does not store all entries in the table $D_v$ but only a fraction $t$ of these entries. The issue is to select which entries to keep, and which entries to drop. For our scheme to work, we need to guarantee that, if the distance to node $u$ is not stored in~$D_v$, then there is a node $v'$ on a shortest path from $v$ to $u$, at distance at most $t$ from $v$, that stores $\text{dist}(v',u)$ in its table~$D_{v'}$.
	
	We use the randomized hitting-set technique presented in~\cite{UllmanY91}, to prove the existence of a correct assignment of certificates: for every node $u\neq v$, each node $v$ keeps $\text{dist}(v,u)$ in its table $D_v$ with probability $\frac{c\log n}{t}$. In addition, node $v$ systematically keeps $D_v(v)=0$ in its table. 
	Note that our scheme is deterministic, although we use the probabilistic method to prove its existence.
	We derive the following two properties.  
	
	\begin{enumerate}
		\item For every pair of nodes $(u,v)$, let us denote by $P_{v,u}$ the path of length~$t$ formed by the $t$ first nodes on a shortest path from $v$ to $u$, and let $X_{v,u}$ denote the sum of $t$~Bernoulli random variables with parameter $\frac{c\log n}{t}$. By the use of Chernoff bounds (Lemma~\ref{lem:chernoff}), we have $\Pr[X_{v,u}\leq \frac12 c\log n] \leq e^{-\frac{c}{8}\log n}=n^{-c/8}<\frac{1}{2n^2}$ for $c$ large enough. Therefore, by union bound, the probability that there exists a pair of nodes $(u,v)$ such that no nodes of $P_{v,u}$ store the distance to node~$u$ is less than~$\frac12$. 
		
		\item For every node $v$, let $Y_v$ be the number of nodes for which $v$ keeps the distance these nodes. Again, by Chernoff bound (Lemma~\ref{lem:chernoff}), $\Pr[Y_v\geq \frac{2cn\log n}{t}]\leq e^{-\frac{c \, n \log n}{3\,t}}\leq e^{-\frac{c \log n}{3}}=n^{-c/3}<\frac{1}{2n}$ for $c$ large enough. Therefore, by union bound, the probability that there exists a node $v$ that stores the distances to more than $\frac{2cn\log n}{t}$ nodes is less than~$\frac12$. 
	\end{enumerate}
	
	Let $\mathcal{E}_{v,u}$ be the event ``at least one node of $P_{v,u}$ stores its distance to node~$u$'', and let $\mathcal{E'}_{v}$ be the event ``node $v$ stores no more than $\frac{2cn\ln n}{t}$ distances to other nodes''. We derive from the above that 
	\[
	\Pr[\forall (u,v)\in V\times V, \mathcal{E}_{v,u} \wedge \mathcal{E'}_{v}]>0.
	\]
	It follows that there exists an assignment of entries to be kept in each table $D_v, v\in V$, such that each resulting partial table is of size $O\left(\frac{n\log^2n}{t}\right)$ bits, and, for every two nodes $u$ and $v$, at least one node at distance at most~$t$, on a shortest path from $v$ to $u$, stores its distance to node~$u$. 
	
	It remains to show that these sparse certificates can be verified in~$t$ rounds. Let $B(v,t)$ be the ball of radius~$t$ around $v$. Each node $v$ verifies that, first, for every node $v'\in B(v,t)$ such that both $v$ and $v'$ store the distance to a same node~$u$, we have $D_{v'}(u)-\text{dist}(v,v') \leq D_v(u) \leq D_{v'}(u)+\text{dist}(v,v')$, and, second, for every node~$u$ such that $v$~stores its distance to~$u$, there exists a node $v'\in B(v,t)$ such that $D_v(u) = D_{v'}(u)+\text{dist}(v,v')$. Third, using the distances collected in $B(v,t)$, node $v$ constructs the table $D'_v$ where $D'_v(u)=D_v(u)$ if $u$ is stored in $D_v$, and $D'_v(u)=\min_{v'\in B(v,t)} (D_{v'}(u)+\text{dist}(v,v'))$ otherwise. Finally, node $v$ checks that $D'_v(u)\leq x(v)$ for every node $u$. If all these tests are passed, then $v$ accepts, otherwise it rejects. 
	
	By the setting of the partial tables $D_v, v\in V$, in a legal instance, we get that $D'_v(u)=\text{dist}(v,u)$ for every node~$u$, and therefore all nodes accept. Instead, if there exists $(u,v)\in V\times V$ such that information about $u$ is stored in $D_v$, but $D_v(u)\neq \text{dist}(u,v)$, then let us consider such a pair $(u,v)$ where $D_v(u)$ is minimum. For $v$ to accept, there must exist some node $v'\in B(v,t)$ such that $D_{v'}(u)=D_v(u)-\text{dist}(v,v')$. By the choice of the pair $(u,v)$, $D_{v'}(u)=\text{dist}(u,v')$, and thus $D_v(u)= \text{dist}(u,v)$, a contradiction. Therefore $v'$ cannot exist, and thus $v$ rejects. It follows that this scheme is a correct $t$-PLS for diameter, using $O\left(\frac{n}{t}\log^2 n\right)$-bit certificates. 
\end{proof}

\begin{lemma}
	\label{lemma: pls diam cannot scale better than linear}
	For every $t\in [1,n/\log n]$,
	$\pls(\mbox{\sc diam},t)=\Omega\left(\frac{n}{t\log n}\right)$.
\end{lemma}

\begin{figure}[t]
	\begin{center}
		\includegraphics[scale=1,
		trim=3cm 21.3cm 4.5cm 3cm,clip]{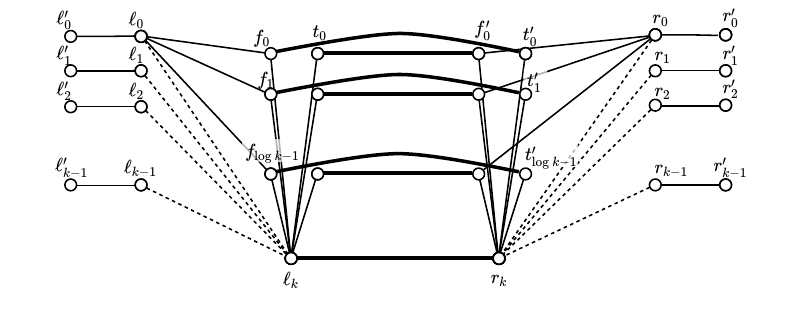}
		\caption{
			The lower bound graph construction.
			Thin lines represent $P$-paths,
			thick lines represent $(2t+1)$-paths.
			The dashed lines represent edges whose existence depend on the input (these edges are not considered in Claims~\ref{claim: lb graph distances between central nodes} and~\ref{claim: lb graph L R different index}).
			The paths connecting~$\ell_i$ and~$r_i$ to their binary representations
			are omitted, except for those of~$\ell_0$ and~$r_0$.}
		\label{fig: diam lb graph}
	\end{center}
\end{figure}

The proof of this lemma is based on two previous works.
First, Peleg and Rubinovich~\cite{PelegR99} proved distributed lower bounds by a reduction to communication complexity, and
second, Abboud et al~\cite{AbboudCHK16} gave a lower bound on the approximability
of the diameter, using a simpler reduction.
Both results are in  the \CONGEST{} model.
We adapt the graph construction from the latter and the simulation technique from the former, in order to prove lower bounds on the label size in the $t$-PLS model.
Interestingly, in our model there are no bandwidth restrictions,
which makes it more similar to the \LOCAL{} model
than to the \CONGEST{} model,
yet the constructions we use were designed to achieve \CONGEST{} lower bound.

We now describe the construction of the lower bound graph (see Figure~\ref{fig: diam lb graph}).
Let $k=\Theta(n)$ be a parameter whose exact value will follow from the graph construction.
Alice and Bob use the graph in order to decide \disj{} on $k$-bit strings.
Let $t$ be the parameter of the $t$-PLS,
which may or may not be constant,
and let $P$ be a constant, $1\leq P\leq t$, to be chosen later (for the proof of Lemma~\ref{lemma: pls diam cannot scale better than linear}, $P=1$ suffices, but we will reuse the construction for spanners later, and then we will need to set $P$ to a precise value).
The graph consists of the following sets of nodes:
$L= \set{\ell_0,\ldots,\ell_{k-1}}$,
$L'=\set{\ell'_0,\ldots,\ell'_{k-1}}$,
$T= \set{t_0,\ldots,t_{\log k -1}}$,
$F= \set{f_0,\ldots,f_{\log k -1}}$, and
$\ell_k$,
which will be simulated by Alice,
and similarly
$R= \set{r_0,\ldots,r_{k-1}}$,
$R'=\set{r'_0,\ldots,r'_{k-1}}$,
$T'=\set{t'_0,\ldots,t'_{\log k -1}}$,
$F'=\set{f'_0,\ldots,f'_{\log k -1}}$, and
$r_k$,
which will be simulated by Bob.

The nodes are connected by paths,
where the paths consist of additional, distinct nodes.
For each $0\leq i\leq k-1$,
connect by $P$-paths (i.e., paths of $P$ edges and $P-1$ new nodes) the pairs of nodes:
$(\ell_i, \ell'_i)$
and $(r_i, r'_i)$.
Add such paths also between
$\ell_{k}$ and all $t_h\in T$ and $f_h\in F$,
and between $r_{k}$ and all $t'_h\in T'$ and $f'_h\in F'$.
Connect by a $P$-path each $\ell_i \in L$ with the nodes representing its binary encoding,
that is, connect $\ell_i$ to each $t_h$ that satisfies $i[h]=1$,
and to each $f_h$ that satisfies $i[h]=0$,
where $i[h]$ is bit $h$ of the binary encoding of $i$.
Add similar paths between each $r_i\in R$
and its encoding by nodes $t'_h$ and $f'_h$.
In addition,
for each $0\leq h \leq \log k -1$,
add a $(2t+1)$-path from $t_h$ to $f'_h$ and from $f_h$ to $t'_h$,
and a similar path from $\ell_{k}$ to $r_{k}$.

We start by noting that the distances between most pairs of nodes is at most $4P+2t+1$. (The only nodes not considered in the following claim are the ones of~$R'$ and~$L'$.)

\begin{claim}
	\label{claim: lb graph distances between central nodes}
	The distance between every two nodes in $L\cup R\cup F\cup T\cup F'\cup T'\cup\{\ell_k,r_k\}$ and in the paths connecting these nodes is at most
	$4P+2t+1$.
\end{claim}

\begin{proof}
	The proof of this claim is by a lengthy but simple case analysis. 
	
	The distance from $\ell_k$ to every node in $F\cup T$ is $P$, and to every node in $F'\cup T'$ is $P+2t+1$, through $r_k$ or through the corresponding node in $F\cup T$.
	Hence, the distance from $\ell_k$ to every node in $F\cup T\cup F'\cup T'\cup\{\ell_k,r_k\}$ or in the paths between them is at most $P+2t+1$. 
	The same holds for distances from $r_k$.
	
	The distance from a node in $L$ to a node in 
	$F\cup T\cup F'\cup T'\cup\{\ell_k,r_k\}$ or in the paths between them is at most $3P+2t+1$, by a path connecting it to $\ell_k$ in $2P$ hops, and then to the target node. 
	If a node $u$ resides on a path from a node in $L$ to a node in $F\cup T$, then the distance from $u$ to any node in $F\cup T\cup F'\cup T'\cup\{\ell_k,r_k\}$ or in the paths between them is less than $3P+2t+1$, by a similar argument.
	
	For a node in $R$, or in the paths from nodes in $R$ to $F'\cup T'$, the distance to any node in $F\cup T\cup F'\cup T'\cup\{\ell_k,r_k\}$ is also at most $3P+2t+1$, by a similar argument.

	The distance between two nodes in $L$ is at most $4P$, going through $\ell_k$, and the same holds for nodes on path connecting $L$ and $F\cup T$.
	The distance between two nodes in $R$, or in the paths from $R$ to $F'\cup T'$,
	is similarly bounded, by a path through $r_k$.
	The distance between a node in $L$ and a node in $R$ is at most $4P+2t+1$, through $\ell_k$ and $r_k$, and the same holds for nodes on paths connecting $L$ and $F\cup T$ and $R$ to $F'\cup T'$.
	
	Finally, for two nodes in $F\cup T\cup F'\cup T'\cup\{\ell_k,r_k\}$ or in the paths between them, note that they reside on a $(4P+4t+2)$-cycle through $\ell_k$ and $r_k$,
	and thus their distance is at most $2P+2t+1$.
\end{proof}

\begin{claim}
	\label{claim: lb graph L R different index}
	The distance between two nodes in $L'$ is at most $6P$, and so is the distance between two nodes in $R'$. The same holds for two nodes on the path connecting $L'$ and $L$ (including nodes in $L'$), or on the path connecting $R'$ and $R$.
	
	The distance between two nodes $\ell'_i\in L'$ and $r'_j\in R'$ with $i\neq j$ is at most $4P+2t+1$.
	The same bound applies to the distance between a node on the path connecting $\ell'_i$ to $\ell_i$ and a node on one of the $P$-path going out of $r_j$ (to nodes in $\{r'_j\}\cup F'\cup T'$),
	and to the distance between a node on the path connecting $r'_j$ to $r_j$, and a node on one of the $P$-path going out of $\ell_i$ (to nodes in $\{\ell'_i\}\cup F\cup T$).
\end{claim}

\begin{proof}
	Each node in $L'$ is at distance $3P$ from $\ell_k$, and nodes on paths from $L'$ to $L$ are closer to $\ell_k$.
	The same holds for nodes in $R'$ or in the paths from $R'$ to $R$, with $r_k$.
	This proves the first claim.
	
	For the second claim, consider two nodes $\ell_i$ and $r_j$, $i\neq j$,
	and an index $h$ such that $i[h]\neq j[h]$.
	Assume without loss of generality that $i[h]=1$ while $j[h]=0$.
	There is a simple path connecting these nodes, through
	$\ell_i, t_h, f'_h, r_j$,
	and its length is $2P+2t+1$.
	The other nodes discussed are at distance at most $P$ from $\ell_i$ or from $r_j$, 
	and the claim follows using a similar path.
\end{proof}

We now introduce the link with communication complexity, in order to use Theorem~\ref{thm:non-det-disj}.
Assume Alice and Bob want to solve the \disj{} problem
for two $k$-bit strings $S_A$ and $S_B$,
by using a non-deterministic protocol.
They build the graph described above,
and add the following edges:
$(\ell_i,\ell_{k})$ whenever $S_A[i]=0$, and
$(r_i,r_{k})$ whenever $S_B[i]=0$. (Note that these are simply edges, not paths.)
The next claim is at the heart of our proof.

\begin{claim}
	\label{claim: lb graph disj vs diam}
	If $S_A$ and $S_B$ are disjoint then $D= 4P+2t+2$,
	and otherwise $D\geq 6P+2t+1$.
\end{claim}

\begin{proof}
	First, note that adding edges can only shorten distances, so claims~\ref{claim: lb graph distances between central nodes} and~\ref{claim: lb graph L R different index} hold after adding the input-dependent edges.	
	These lemmas show that the distances between almost all pairs of nodes is at most $4P+2t+1$, with the exception of pairs of the form $\ell'_i\in L'$ and $r'_i\in R'$, for some $i$, and pairs where one node is on a path connecting $\ell'_i$ and $\ell_i$, and the other on a path connecting $r'_i$ and $r_i$. 
	Hence, we need to focus only on pairs of nodes of the form $(\ell'_i,r'_i)$.	
	
	If the inputs are disjoint then for every $0\leq i\leq k-1$,
	either $S_A[i]=0$ or $S_B[i]=0$, so at least one of the edges $(\ell_i,\ell_{k})$ and $(r_i, r_{k})$ exists. Thus, there is a $(4P+2t+2)$-path between $\ell'_i$ and $r'_i$,
	through $\ell_{k}$ or through $r_{k}$, e.g. $\ell_i',\ell_i,\ell_{k},r_{k},t'_h,r_i,r_i'$ for the case $S_A[i]=0$.
	The other distances in the graph are not larger,
	and the distance between $\ell'_i$ and $r'_i$ is exactly $4P+2t+2$
	whenever $S_A[i]\neq0$ or $S_B[i]\neq0$
	\footnote{In the case where both inputs are the 0 strings, the diameter is only $4P+2t+1$;
	we can exclude this unique case by forbidding this input to the \disj{} problem without changing its asymptotic complexity.}.
	The other nodes discussed are on this path, so their distances are smaller.
	Hence, the diameter in this case is exactly $4P+2t+2$.
	
	On the other hand, if the sets are not disjoint,
	there is an index $i$ such that $S_A[i]=S_B[i]=1$.
	Hence, there is no edge connecting $\ell'_i$ and $\ell_k$, and no edge connecting $r'_i$ and $r_k$. 
	We can split the nodes of $F\cup T$ to two sets: nodes that correspond to the binary representation of $i$, and have distance $2P$ from $\ell'_i$; and the other nodes, whose distances from $\ell'_i$ are $4P$.
	A similar split of $F'\cup T'$ applies, regarding the distance from $r'_i$. 
	The graph is specifically crafted so that the nodes close to $\ell'_i$ and the nodes close to $r'_i$ are far apart from one another. 
	More precisely, if $\ell_i'$ is at distance $2P$ from a node $a$ in $F\cup T$, and $r_i'$ is at distance $2 P$ from a node $b$ in $F'\cup T'$, then $a$ and $b$ are not linked by a path on $2t+1$ nodes. Then a shortest path from $\ell_i'$ to $r_i'$ should use four $P$-paths, either on the left part of the graph, or on the right part; and then two additonal $P$-paths on the other part.  
	Hence, any path connecting $\ell'_i$ and $r'_i$ has length at least $6P+2t+1$, and the claim follows.
	\end{proof}

With this claim in hand, we are ready to prove
Lemma~\ref{lemma: pls diam cannot scale better than linear}.

\begin{proof}[Proof of Lemma~\ref{lemma: pls diam cannot scale better than linear}.]
	Fix $t\in[1,n/\log n]$,
	and let $S_A$ and $S_B$ be two input strings for the \disj{} problem
	on $k$ bits.
	We show how Alice and Bob can solve \disj{} on $S_A$ and $S_B$
	in a nondeterministic manner,
	using the graph described above
	and a $t$-PLS for $\mdiam= 4P+2t+2$.

\begin{figure}[t]
	\begin{center}
		\includegraphics[scale=1,
		trim=3cm 21.5cm 5.5cm 2.5cm,clip]{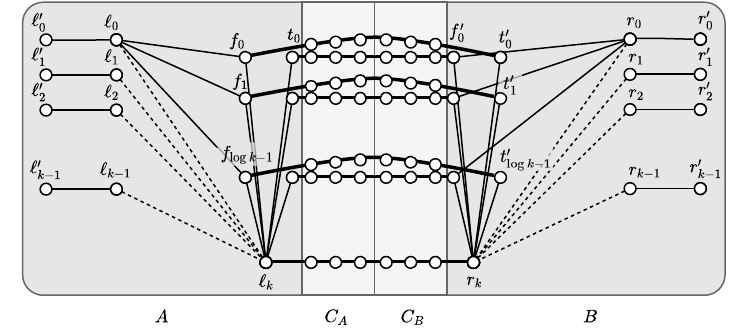}
		\caption{
			The lower bound graph construction for $t=3$,
			and the partition of $V$ into $A,C_A,C_B$ and $B$.}
		\label{fig: diameter lb simulation graph}
	\end{center}
\end{figure}
	Let us first define a partition of the nodes of the graph (see Figure~\ref{fig: diameter lb simulation graph}).
	Let $A$ be the set of nodes that are either in $L\cup L' \cup T \cup F \cup\set{\ell_{k}}$ or in a $P$-path between two nodes of this union.
	Similarly let $B$ be the set of nodes that are either in $R\cup R' \cup T' \cup F' \cup\set{r_{k}}$  or in a $P$-path between two nodes of this union. 
	For each pair of nodes $(a,b)\in A\times B$ that are connected by a $(2t+1)$-path,
	let $P_{ab}$ be this path,
	and $\set{P_{ab}(i)}$, $i=0,\ldots, 2t+1$ be its nodes in consecutive order,
	where $P_{ab}(0)=a$ and $P_{ab}(2t+1)=b$.
	Let $C$ be the set of nodes on such paths, at the exception of the endpoints. 
	The set $C$ is partitioned into $C_A$ and $C_B$, where $C_A$ is the set of nodes $\set{P_{ab}(i)}$, $i=1,\ldots, t$, and $C_B$ is the set of nodes $\set{P_{ab}(i)}$, $i=t+1,\ldots, 2t$. 
	Note that $A\cup B\cup C_A \cup C_B$ forms a partition of the nodes of the graph.
	
	Alice interprets her nondeterministic string as the certificates given
	to the nodes in $A\cup C$,
	and sends the certificates of $C$ to Bob.
	Bob interprets his nondeterministic string as the certificates of $B$,
	and gets the certificates of $C$ from Alice.
	Alice can simulate the verifier on the nodes of $A\cup C_A$, as follows. 
	Consider a node $v$ of $A\cup C_A$. Our construction assures that the ball of radius $t$ around $v$ is contained in $A\cup C$. Hence, Alice knows the topology of this ball --- the only edges in the graph she is not aware of are between $r_k$ and nodes in $R$. In addition, Alice knows all the certificates of the nodes in this ball.
	
	Thus, Alice can simulate the verifier on every node $v$ of $A\cup C_A$. Similarly, Bob can simulate the verifier on every node of $B\cup C_B$, because he also knows the full topology of the $t$-balls around these nodes, and the certificates of the ball's nodes.
	
	Using this simulation,
	Alice and Bob can determine whether \disj\ on $(S_A,S_B)$ is true:
	from Claim~\ref{claim: lb graph disj vs diam},
	we know that if it is true then $\mdiam= 4P+2t+2$,
	and the verifier of the PLS accepts,
	while otherwise it rejects.
	The nondeterministic communication complexity of the true case of
	\disj{} on $k$-bit strings is $\Omega (k)$ by Theorem~\ref{thm:non-det-disj},
	so Alice and Bob must communicate this many bits.
	In our construction, $k$ is in $\Theta(n)$, thus Alice and Bob have to communicate $\Omega(n)$ bits.
	From the graph definition,
	$|C|=\Theta(t\log n)$
	which implies
	$\pls(\mbox{\sc diam},t)=\Omega\left(\frac{n}{t\log n}\right)$,
	as desired.
\end{proof}

Let $\alpha, \beta$ be two non-negative integers, that may depend on $n$. 
For any labeled graph~$(G,x)$, $\mbox{\sc $\beta$-additive spanner}$ is the predicate on $(G,x)$ that states whether, for every $v\in V(G)$,  $x(v)\subseteq \{\id(w), w\in N(v)\}$ for every $v\in V(G)$, and whether the collection of edges $E_H=\{\{v,w\}, v\in V(G), w\in x(v)\}$ forms a  $\beta$-additive spanner of $G$, i.e., a subgraph $H$ of $G$ such that, for every two nodes $s,t$, we have $\dist_H(s,t)\leq \dist_G(s,t)+\beta$. Similarly, the predicate $\mbox{\sc $(\alpha,\beta)$-spanner}$ on $(G,x)$ states whether $H$ is such that, for every two nodes $s,t$, $\dist_H(s,t)\leq \alpha \, \dist_G(s,t)+\beta$.
There is a proof-labeling scheme for $\mbox{\sc $(\alpha,\beta)$-spanner}$ that weakly scales linearly.
For the spacial case of $\alpha=1$,
i.e., the predicate $\mbox{\sc $\beta$-additive spanner}$,
this is also optimal if $\beta$ is constant.

\begin{theorem}
	\label{thm: spanners scaling}
	For every pair of integers $\alpha\geq1$ and $\beta\geq0$, there is a one-round proof-labeling scheme for $\mbox{\sc $(\alpha,\beta)$-spanner}$ of size $O(n)$,

	and a $t$-round proof-labeling scheme of size $O(\frac{n}{t}\log^2 n)$.
	Moreover, for a constant $\beta$ we also have
	$\pls(\mbox{\sc $\beta$-additive spanner},t)=\Omega(\frac{n}{t})$.

\end{theorem}


\begin{proof}
	
	The upper bound proof is similar to the one of Lemma~\ref{lemma: pls diam scales linearly} for diameter. Specifically, instead of using a unique table $D_v$ storing distances in the graph, every node~$v$ stores two tables $D_v$ and $\widehat{D}_v$, where $\widehat{D}_v$ stores the distances in the spanner. To scale, each node keeps only a fraction $\frac{c\log n}{t}$ of the entries in each table, for some constant $c>0$. For $c$ large enough, this is sufficient for every node to recover its distance to every other node in both the graph and the spanner, by using the same arguments as in the proof of Lemma~\ref{lemma: pls diam scales linearly}. The verification again proceeds as in the proof of Lemma~\ref{lemma: pls diam scales linearly}, excepted that, instead of checking whether $D_v(u)\leq x(v)$ for every node~$u$, every node~$v$ checks that, for every node~$u$, $\widehat{D}_v(u)\leq \alpha D_v(u) +\beta$ or $\widehat{D}_v(u)\leq D_v(u) + \beta$, depending on whether one is dealing with general spanners or additive spanners, respectively. Correctness directly follows from the same arguments as in the proof of Lemma~\ref{lemma: pls diam scales linearly}.

	To achieve the lower bound,
	we rely on the same construction as in the proof of
	Lemma~\ref{lemma: pls diam cannot scale better than linear}.
	The base graph, for which a spanner is constructed, is the graph described above for the case where both $S_A$ and $S_B$ are the all-$0$ strings, i.e.\ all the potential edges are present in the graph.
	To decide \disj{}, Alice and Bob build a spanner for this base graph in the same manner as the graph is constructed in the proof of Lemma~\ref{lemma: pls diam cannot scale better than linear}, i.e., the spanner contains all the nodes and paths, and among the input-dependent edges, the ones that  correspond to zeros in the input strings.
	This construction guarantees that each graph edge that is not in the spanner
	has a $2P$-path between its endpoint that is in the spanner.
	If the inputs are disjoint,
	then at most one edge in each shortest path is not in the spanner,
	so the stretch is at most $2P-1$.
	On the other hand, if the inputs are not disjoint then
	there is a pair $(\ell'_i,r'_i)$ that stretches from
	$4P+2t+1$ in the original graph
	to $6P+2t+1$ in the spanner.
	
	We pick $P$ such that the spanner is a legal $(\alpha,\beta)$-spanner
	if and only if the inputs are disjoint,
	that is, $4P+2t+2 \leq \alpha(4P+2t+1) +\beta$,
	while $6P+2t+1 > \alpha(4P+2t+1) +\beta$.
	For additive spanners, i.e.\ $\alpha=1$, the first condition always holds;
	to guarantee the second, we choose $P>\beta/2$.
	The lower bound is thus $\Omega\left(\frac{n}{\beta t}\right)$, or $\Omega\left(\frac{n}{t}\right)$ for a constant $\beta$.
\end{proof}

\section{Distributed Proofs for Spanning Trees}
\label{sec: spanning trees}
In this section, we study two specific problems which are classical in the domain of
proof-labeling schemes: the verification of a spanning tree, and of a minimum-weight spanning tree.
The predicates $\ST$ and $\MST$ are the sets of labeled graphs where some edges are marked and these edges form a spanning tree, and a minimum spanning tree, respectively. For these predicates, we present proof-labeling schemes that scale linearly in $t$.
Note that $\ST$ and $\MST$  are problems on general labeled graphs and not on trees, i.e., the results in this section improve upon Section~\ref{sec:optimal-uniform} (for these specific problems),
and are incomparable with the results of Section~\ref{sec:scaling-on-trees}.

Formally, let $\cF$ be the family of all connected undirected, weighted, labeled graphs $(G,x)$.
Each label $x(v)$ contains a (possibly empty) subset of edges adjacent to $v$,
which is consistent with the neighbors of $v$, and we denote the collection of edges represented in $x$ by $T_x$.
In the $\ST$ (respectively, $\MST$) problem, the goal is to decide for every labeled graph $(G,x)\in \cF$ whether $T_x$ is a spanning tree of $G$ (respectively, whether $T_x$ is a spanning tree of $G$  with the sum of all its edge-weights minimal among all spanning trees of $G$).
For these problems we have the following results.

\begin{theorem}\label{thm:ST}
	For every $t\in O(\log n)$, we have that $\pls(\ST,t) = O\left(\frac{\log n}{t}\right)$.
\end{theorem}

Ostrovsky et al.~\cite[Theorem 8]{OstrovskyPR17} designed a radius-$t$ proof-labeling scheme for acyclicity,
with $(\ceil*{\log n/t})$-bit certificates, for $t\leq D$.
Here, $D$ is the diameter of the graph, which is at least the largest depth of a tree in it.
In the scheme, each tree is oriented outwards from an arbitrarily chosen root, and
after running the verification process, each node knows who is his parent in its tree,
and the root of each tree knows it is the root.

For completeness we describe here in a high level this $t$-PLS for acyclicity.
The following scheme can be used to verify that a graph contains no cycles using certificates of size $O(\log n)$ in a single round. The certificate of a node $v$ consists of an integer $d(v)$ which encodes the distance from $v$ to a root (which has $d(v) = 0$). Nodes verify the correctness of the certificates in a single communication round. If $v$ satisfies $d(v) = 0$ (i.e., $v$ is a root), then it accepts the certificate if all of its neighbors $w$ satisfy $d(w) = 1$. If $v$ satisfies $d(v) \not= 0$ then $v$ verifies that $v$ has exactly one neighbor $u$ with $d(u) = d(v) - 1$ while all other neighbors $w$ satisfy $d(w) = b(v) + 1$. This scheme is used, for example, in~\cite{APV91,IL94,AO94}.
To achieve certificates of size $O((\log n) / t)$ for acyclicity, Ostrovsky et al. simulate the $1$-PLS  described above. The idea is to take only some ``special'' nodes, break their $O(\log n)$-bit  certificates  indicating the distances into shares of size $O((\log n)/t)$, and spread the shares on paths of length $\Theta(t)$ down the tree (similarly to the way we spread out the certificates of the border and extra-border nodes in the proof of Lemma~\ref{lem:spreading}).
In the verification process, the ``special'' nodes recover their distances to the root by collecting the relevant shares from their $t$-neighborhood, and all other nodes recover their distances using the fact that a node with distance $d_1$ (down the tree) to a node with distance $d_2$ to the root, must be with distance $d_1+d_2$ to the root. Then, nodes can make the verification exactly as in the $1$-PLS described above.
This scheme plays an essential role in the proof of Theorem~\ref{thm:ST}.

\begin{proof}[Proof of Theorem~\ref{thm:ST}]
	To prove that a marked subgraph $T_x$ is a spanning tree,
	we need to verify it has the following properties: (1) spanning the graph, (2) acyclic, (3) connected.
	We choose an arbitrary node as the root of $T_x$.
	
	The certificate of a node $v$ is composed of three parts:
	a bit indicating whether the tree is shallow as in the proof of Theorem~\ref{thm:tree-scaling},
	$O(\log n/t)$-bits by the scheme of Ostrovsky et al.~\cite[Theorem 8]{OstrovskyPR17} for acyclicity,
	and $O(\log n/t)$-bits that are a part of the \ID{} of the root, as in the proof of Theorem~\ref{thm:tree-scaling}.
	
	In the verification process,
	the nodes first verify they all have the same first bit. If the tree is shallow, all nodes but the root accept,
	and the root collects the whole structure of the graph and of $T_x$ and verifies it is a spanning tree.
	Otherwise, each node verifies that at least one of its edges is marked to be in $T_x$, making sure $T_x$ is spanning all the graph.
	The nodes then run the verification process from~\cite[Theorem 8]{OstrovskyPR17}, while ignoring edges not in $T_x$, to make sure the graph is acyclic.
	Finally, they all run the reconstruction process from Theorem~\ref{theo:universal} to find the root \ID{}, and the root of $T_x$ (as defined by the acyclicity scheme) verifies this root \ID{} is indeed its own \ID{}. This guarantees $T_x$ is a connected forest, i.e. a tree, as desired.
\end{proof}

\begin{theorem}\label{thm:MST}
	For every $t\in O(\log n)$, we have that $\pls(\MST,t) = O\left(\frac{\log^2 n}{t}\right)$.
\end{theorem}

Our theorem only applies for $t\in O(\log n)$, meaning that we can get from proofs of size $O(\log^2 n)$ to proofs of size $O(\log n)$, but not to a constant.
For the specific case $t=\Theta(\log n)$, our upper bound matches the lower bound of Korman et al.~\cite[Corollary 3]{KormanKM15}.
In the same paper, the authors also present an $O(\log^2n)$-round verification scheme for $\MST$ using $O(\log n)$ bits of memory at each node (both for certificates and for local computation). Removing the restriction of $O(\log n)$-bit memory for local computation, one may derive an $O(\log n)$-round verification scheme with $O(\log n)$ proof size out of the aforementioned $O(\log^2n)$-round scheme, which matches our result for $t=\Theta(\log n)$. The improvement we present is two-fold: our scheme is scalable for different values of $t$ (as opposed to schemes for only $t=1$ and $t=\Theta(\log n)$), and our construction is much simpler, as described next.

Our upper bound is based on a classic $1$-round PLS for MST~\cite{KormanKM15,KormanK07},
which in turn builds upon the algorithm of Gallager, Humblet, and Spira (GHS)~\cite{gallager83mst} for a distributed construction of an MST.
The idea behind this scheme is, given a labeled graph $(G,x)$,  to verify that $T_x$ is consistent with an
execution of the GHS algorithm in $G$. 

The GHS algorithm maintains a spanning forest that is a subgraph of the minimum spanning tree,
i.e., the trees of the forest are fragments of the desired minimum spanning tree.
The algorithm starts with a spanning forest consisting of all nodes and no edges.
At each phase each of the fragments
adds the minimum-weight edge going out of it, thus merging several fragments into one.
After $O(\log n)$ iterations, all the fragments are merged into a single component,
which is the desired minimum-weight spanning tree.
We show that each phase can be verified with $O(\log n / t)$ bits, giving a total complexity of $O(\log^2 n / t)$ bits.

We note that the GHS algorithm requires either distinct edge weights or a fixed ordering of the edges in order to break ties between edges with the same weight in a consistent way. 
In addition, given a labeled graph $(G,x)$, any MST $T$ can be a result of the GHS algorithm on $(G,x)$ for some fixed ordering of the edges: any order that prefers the edges of $T$ over the other edges will result in $T$ as an output.

\begin{proof}
	Let $(G,x)$ be a labeled graph such that $T_x$ is a minimum-weight spanning tree,
	and fix an ordering of the edges such that an execution of the GHS algorithm on $(G,x)$ yields the tree $T_x$. Let $\POS(e)$ be the position of $e$ in this ordering.
	If $t$ is greater than the diameter $D$ of $G$, every node can see the entire labeled graph in the verification process, and we are done;
	we henceforth assume $t\le D$.
	The certificates consist of four parts.
	
	First, we choose a root and orient the edges of $T_x$ towards it.
	We give each node its distance from the root modulo $3$,
	which allows it to obtain the \ID{} of its parent and the edge pointing to it in one round.
	Second, we assign the  certificate described above for $\ST$ (Theorem~\ref{thm:ST}),
	which certifies that $T_x$ is indeed a spanning tree. This requires $O(\log n/t)$~bits.
	
	The third part of the certificate tells each node 
	the position of the edge connecting it to its parent in the fixed ordering, the phase in which the edge is added to the tree in the execution of the GHS algorithm,
	and which of the edge's endpoints added it to the tree.
	Note that after one round of verification, each node knows for every incident edge, its position in the ordering, at which phase it is added to the spanning tree,  and by which of its endpoints.
	This part uses $O(\log n)$ bits.

	The fourth part of the certificate consists of  $O(\log^2 n/t)$ bits,  $O(\log n/t)$ for each of the $O(\log n)$ phases of the GHS algorithm.
	To define the part of a certificate of every phase, fix a phase, a fragment $F$ in the beginning of this phase, and 
	let $e=(u,v)$ be the edge added to the spanning tree by fragment $F$ in this phase, where $u\in F$ and $v\notin F$.
	Our goal is that the nodes of $F$ verify together that $e$ is a minimum-weight outgoing edge with the smallest position in the ordering, and that no other edge was added by~$F$ in this phase.
	To this end, we first
	orient the edges of $F$ towards $u$, i.e., set $u$ as the root of~$F$.
	If the depth of $F$ is less than $t$, then in $t-1$ rounds the root $u$ can see all of $F$ and check that $(u,v)$ is the desired  minimum-weight outgoing edge. All other nodes just have to verify that no other edge is added by the nodes of $F$ in this phase.
	Otherwise, if the depth of $F$ is at least~$t$, by Theorem~\ref{theo:universal}, the information about $\ID(u)$, $w(e)$, and $\POS(e)$ can be spread on $F$ such that in $t$ rounds it can be collected by all nodes of $F$. With this information known to all the nodes of $F$, the root can locally verify that it is named as the node that adds the edge and that it has the named edge with the right weight and position. The other nodes of $F$ can locally verify that they do not have incident edges with a smaller weight, 
	that every incident edge with the same weight has a greater position, 
	and that no other edge is added by~$F$.
	This part takes $O(\log n/t)$ bits per iteration, which sums to a total of $O(\log^2 n/t)$ bits.
	
	Overall, our scheme verifies that $T_x$ is a spanning tree, and that it is consistent with every phase of some execution of the GHS algorithm. Therefore, the scheme accepts $(G,x)$ if and only if $T_x$ is a minimum spanning tree.
\end{proof}

\section{Open problems}

We finish with a list of open problems. We first present a series of open questions directly related to our setting, and then consider natural alternative settings.

We have proved that, for many classical Boolean predicates on labeled graphs (including MST), there are proof-labeling schemes that  scale linearly with the radius of the scheme, i.e., the number of rounds of the verification procedure. More generally, we have shown that for \emph{every} Boolean predicate on  labeled trees, cycles and grids, there is a proof-labeling scheme that scales linearly with the radius of the scheme. This yields the following question:

\begin{problem}
	Prove or disprove that, for every predicate $\mathcal{P}$ on labeled graphs, there is a proof-labeling scheme for $\mathcal{P}$ that (weakly) scales linearly.
\end{problem}

In fact, the scaling factor might even be larger than~$t$, and be as large as $b(t)$ in graphs with ball growth~$b$. We have proved that the uniform part of any proof-labeling scheme can be scaled by such a factor~$b(t)$ for $t$-PLS. This yields the following, stronger open problem:

\begin{problem}
	Prove or disprove that, for every predicate $\mathcal{P}$ on labeled graphs, there is a proof-labeling scheme for $\mathcal{P}$ that scales with factor $\Omega(b(t)\poly\log n)$ in graphs with ball growth~$b$.
\end{problem}

We are tempted to conjecture that the answer to the first problem is positive (as it holds for trees and cycles). However, we believe that the answer to the second problem might well be negative. In particular, it seems challenging to design a proof-labeling scheme for \textsc{diam} that would scale with the size of the balls. Indeed, checking diameter is strongly related to checking shortest paths in the graph, and this significantly restricts the way the certificates can be redistributed among nodes in a ball of radius~$t$. Yet, there might be some other way to certify \textsc{diam}, so we let the following as an open problem:

\begin{problem}
	Is there a proof-labeling scheme for \textsc{diam} that scales by a factor greater than~$t$ in all graphs where $b(t)\gg t$?
\end{problem}

Another type of open problem is to adapt our result to different settings, and in particular, more practical settings. 
Proof-labeling schemes were originally designed as a component of self-stabilizing algorithms, that are algorithms that not only check the certificates but also build the desired object and its certificates. 
Given that we use the probabilistic method in various places a natural question is:

\begin{problem}
Can we have simple explicit construction of the certificate assignments used in this paper? And further, how can these be efficiently implemented in a distributed setting?
\end{problem} 

Even if we consider only the verification phase, we could try to move to a more practical setting. For example for classic self-stabilizing algorithms, the size of the certificates is a measure of the bandwidth needed for certification, which is no longer the case when aggregating certificates from distance greater than one. Indeed, we use a view of radius $t$, which, if we had to send messages, would require either large messages, or many more that $t$ rounds. Thus, a natural open problem is:

\begin{problem}
Can we make the verification algorithm work efficiently in a message-passing model? More precisely, can we perform the verification in $t$ rounds by allowing the nodes to only exchange messages whose size is the same as the certificate size?
\end{problem}  

\subparagraph{Acknowledgements:}
We thank Seri Khoury and Boaz Patt-Shamir for valuable discussions,
and the anonymous reviewers of DISC 2018 and Distributed Computing journal for their comments.

\bibliographystyle{plain}
\bibliography{tPLSbib}
\end{document}